\documentclass[11pt]{article}

\usepackage{amsthm}
\usepackage{graphicx} 
\usepackage{array} 

\usepackage{amsmath, amssymb, amsfonts, verbatim}
\usepackage{hyphenat,epsfig,subcaption,multirow}
\usepackage{nicefrac}
\usepackage{paralist}

\usepackage[linesnumbered,vlined]{algorithm2e}
\usepackage{color-edits} 
\usepackage{xfrac}

\usepackage[font=small,labelfont=bf]{caption}

\usepackage[usenames,dvipsnames,table]{xcolor}

\DeclareFontFamily{U}{mathx}{\hyphenchar\font45}
\DeclareFontShape{U}{mathx}{m}{n}{
      <5> <6> <7> <8> <9> <10>
      <10.95> <12> <14.4> <17.28> <20.74> <24.88>
      mathx10
      }{}
\DeclareSymbolFont{mathx}{U}{mathx}{m}{n}
\DeclareMathSymbol{\bigtimes}{1}{mathx}{"91}

\usepackage{tcolorbox}
\tcbuselibrary{skins,breakable}
\tcbset{enhanced jigsaw}

\usepackage[normalem]{ulem}
\usepackage[compact]{titlesec}

\definecolor{DarkRed}{rgb}{0.5,0.1,0.1}
\definecolor{DarkBlue}{rgb}{0.1,0.1,0.5}

\usepackage{nameref}
\definecolor{ForestGreen}{rgb}{0.1333,0.5451,0.1333}
\definecolor{Red}{rgb}{0.9,0,0}
\usepackage[linktocpage=true,
	pagebackref=true,colorlinks,
	linkcolor=DarkRed,citecolor=ForestGreen,
	bookmarks,bookmarksopen,bookmarksnumbered]
	{hyperref}
\usepackage[noabbrev,nameinlink]{cleveref}
\crefname{property}{property}{Property}
\creflabelformat{property}{(#1)#2#3}
\crefname{equation}{eq}{Eq}
\creflabelformat{equation}{(#1)#2#3}
\crefname{myalgctr}{Algorithm}{Algorithm}
\creflabelformat{myalgctr}{#1#2#3}

\usepackage{bm}
\usepackage{url}
\usepackage{xspace}
\usepackage[mathscr]{euscript}

\usepackage{tikz}
\usetikzlibrary{arrows}
\usetikzlibrary{arrows.meta}
\usetikzlibrary{shapes}
\usetikzlibrary{backgrounds}
\usetikzlibrary{positioning}
\usetikzlibrary{decorations.markings}
\usetikzlibrary{patterns}
\usetikzlibrary{calc}
\usetikzlibrary{fit}
\usetikzlibrary{snakes}

\usepackage{mdframed}

\usepackage[noend]{algpseudocode}
\makeatletter
\def\BState{\State\hskip-\ALG@thistlm}
\makeatother

\usepackage{cite}
\usepackage{enumitem}

\usepackage[margin=1in]{geometry}

\newtheorem{theorem}{Theorem}
\newtheorem{lemma}{Lemma}[section]
\newtheorem{proposition}[lemma]{Proposition}
\newtheorem{corollary}[lemma]{Corollary}
\newtheorem{claim}[lemma]{Claim}
\newtheorem{fact}[lemma]{Fact}

\newtheorem{definition}[lemma]{Definition}

\newtheorem{assumption}[lemma]{Assumption}

\newtheorem*{claim*}{Claim}
\newtheorem*{proposition*}{Proposition}
\newtheorem*{lemma*}{Lemma}
\newtheorem*{problem*}{Problem}

\crefname{lemma}{Lemma}{Lemmas}
\crefname{claim}{Claim}{Claims}
\crefname{introresult}{Result}{Results}

\newtheorem{mdresult}{Result}

\newtheorem{remark}[lemma]{Remark}
\newtheorem{observation}[lemma]{Observation}

\newtheoremstyle{restate}{}{}{\itshape}{}{\bfseries}{~(restated).}{.5em}{\thmnote{#3}}
\theoremstyle{restate}

\theoremstyle{definition}
\newtheorem{mdprotocol}{Algorithm}
\newtheorem{mdprogram}{Program}

\allowdisplaybreaks

\renewcommand{\qed}{\nobreak \ifvmode \relax \else
      \ifdim\lastskip<1.5em \hskip-\lastskip
      \hskip1.5em plus0em minus0.5em \fi \nobreak
      \vrule height0.75em width0.5em depth0.25em\fi}

\newcommand{\Qed}[1]{\qed$_{~\textnormal{\Cref{#1}}}$}

\newcommand{\Leq}[1]{\ensuremath{\underset{\textnormal{#1}}\leq}}

\newcommand{\rs}{Ruzsa-Szemer\'{e}di\xspace}

\newcommand{\Ot}{\ensuremath{\widetilde{O}}}
\newcommand{\eps}{\ensuremath{\varepsilon}}
\newcommand{\Paren}[1]{\Big(#1\Big)}
\newcommand{\Bracket}[1]{\Big[#1\Big]}
\newcommand{\bracket}[1]{\left[#1\right]}
\newcommand{\paren}[1]{\ensuremath{\left(#1\right)}\xspace}
\newcommand{\card}[1]{\vert{#1}\vert}

\newcommand{\IR}{\ensuremath{\mathbb{R}}}

\newcommand{\ceil}[1]{{\left\lceil{#1}\right\rceil}}

\newcommand{\expect}[1]{\Exp\bracket{#1}}

\newcommand{\set}[1]{\ensuremath{\left\{ #1 \right\}}}
\newcommand{\poly}{\mbox{\rm poly}}
\newcommand{\polylog}{\mbox{\rm  polylog}}

\newenvironment{tbox}{\begin{tcolorbox}[
		enlarge top by=5pt,
		enlarge bottom by=5pt,
		 breakable,
		 boxsep=0pt,
                  left=4pt,
                  right=4pt,
                  top=10pt,
                  arc=0pt,
                  boxrule=1pt,toprule=1pt,
                  colback=white
                  ]
	}
{\end{tcolorbox}}

\newcommand{\event}{\ensuremath{\mathcal{E}}}

\newcommand{\rv}[1]{\ensuremath{{\mathsf{#1}}}\xspace}
\newcommand{\rA}{\rv{A}}
\newcommand{\rB}{\rv{B}}
\newcommand{\rC}{\rv{C}}
\newcommand{\rD}{\rv{D}}

\newcommand{\supp}[1]{\ensuremath{\textnormal{\text{supp}}(#1)}}
\newcommand{\distribution}[1]{\ensuremath{\textnormal{dist}(#1)}\xspace}

\newcommand{\II}{\ensuremath{\mathbb{I}}}
\newcommand{\HH}{\ensuremath{\mathbb{H}}}
\newcommand{\mi}[2]{\ensuremath{\def\mione{#1}\def\mitwo{#2}\mireal}}
\newcommand{\mireal}[1][]{
  \ifx\relax#1\relax%
    \II(\mione \,; \mitwo)%
  \else%
    \II(\mione \,; \mitwo\mid #1)%
  \fi
}
\newcommand{\en}[1]{\ensuremath{\HH(#1)}}

\newcommand{\itfacts}[1]{\Cref{fact:it-facts}-(\ref{part:#1})\xspace}

\newcommand{\GRS}{\ensuremath{G^{\textnormal{\textsc{rs}}}}\xspace}
\newcommand{\MRS}{\ensuremath{M^{\textnormal{\textsc{rs}}}}\xspace}

\newcommand{\jstar}{\ensuremath{j^{\star}}}

\newcommand{\rProt}{{\Prot}}

\newcommand{\rJ}{\rv{J}}
\newcommand{\rM}{\rv{M}}

\newcommand{\Gxor}{\ensuremath{G^{\textnormal{\textsc{xor}}}}\xspace}

\newcommand{\GG}{\ensuremath{\mathcal{G}}}
\newcommand{\PP}{\ensuremath{\mathcal{P}}}

\newcommand{\Exor}{\ensuremath{E^{\textnormal{\textsc{xor}}}}}

\newcommand{\eventhide}{\ensuremath{\mathcal{E}_{\textnormal{\textsc{hide}}}}}

\newcommand{\Prot}{{\Pi}}
\newcommand{\prot}{{\pi}}

\renewcommand{\rProt}{\rv{\Prot}}

\newcommand{\rZ}{\rv{Z}}

\newcommand{\rMprot}{\rv{M}_{\prot}}

\newcommand{\rP}{\rv{P}}

\newcommand{\rX}{\rv{X}}
\newcommand{\rY}{\rv{Y}}

\newcommand{\rPxor}{\rv{P}^{\textnormal{\textsc{xor}}}}
\newcommand{\rPrs}{\rv{P}^{\textnormal{\textsc{rs}}}}

\newcommand{\rErsA}{\rv{E}^{\textnormal{\textsc{rs}}}_A}
\newcommand{\rErsB}{\rv{E}^{\textnormal{\textsc{rs}}}_B}

\newcommand{\rMrsA}{\rv{M}^{\textnormal{\textsc{rs}}}_A}
\newcommand{\rMrsB}{\rv{M}^{\textnormal{\textsc{rs}}}_B}

\newcommand{\rZA}{\rv{Z}_A}
\newcommand{\rZB}{\rv{Z}_B}

\newcommand{\EA}{\ensuremath{E^A}}
\newcommand{\EB}{\ensuremath{E^B}}

\title{Beating Two-Thirds For Random-Order Streaming Matching}
\author{Sepehr Assadi\footnote{(\texttt{sepehr.assadi@rutgers.edu}) Department of Computer Science, Rutgers University.  Research supported in by part by the NSF CAREER award CCF-2047061.} \and 
Soheil Behnezhad\footnote{(\texttt{soheil@cs.umd.edu}) Department of Computer Science, University of Maryland. Research supported by Google Ph.D. Fellowship.}}

\date{}

\SetKwRepeat{Do}{do}{while}

\let\Pr\relax
\DeclareMathOperator{\Pr}{Pr}

\renewcommand{\epsilon}[0]{\ensuremath{\varepsilon}}

\let\originalleft\left
\let\originalright\right
\renewcommand{\left}{\mathopen{}\mathclose\bgroup\originalleft}
\renewcommand{\right}{\aftergroup\egroup\originalright}

\makeatletter
\def\thm@space@setup{%
  \thm@preskip= 0.2cm
  \thm@postskip=\thm@preskip 
}
\makeatother

\definecolor{mydarkgreen}{RGB}{20,105,20}
\definecolor{mygreen}{RGB}{20,155,20}
\definecolor{myred}{RGB}{195,20,20}
\definecolor{linkcolor}{RGB}{0,0,230}
\definecolor{mylightgray}{RGB}{230,230,230}
\definecolor{verylightgray}{RGB}{240,240,240}
\definecolor{commentcolor}{RGB}{120,120,120}

\addauthor{sb}{blue}    

\SetCommentSty{mycommfont}

\newcommand{\smparagraph}[1]{
\par\addvspace{0.2cm}
\noindent \textbf{#1}
}

\hypersetup{
     colorlinks=true,
     citecolor= mygreen,
     linkcolor= myred,
     urlcolor= mygreen
}

\usepackage{etoolbox}

\makeatletter
\patchcmd{\@algocf@start}
  {-1.5em}
  {0pt}
  {}{}
\makeatother

\newcommand{\mc}[1]{\ensuremath{\mathcal{#1}}}

\newcounter{myalgctr}

\newenvironment{mytbox}{
\par\addvspace{0.2cm}
\begin{tcolorbox}[width=\textwidth,
                  enhanced,
                  boxsep=2pt,
                  left=1pt,
                  right=1pt,
                  top=4pt,
                  boxrule=1pt,
                  arc=0pt,
                  colback=white,
                  colframe=black,
                  unbreakable
                  ]
}{
\end{tcolorbox}
}

\newenvironment{mytboxh}{
\par\addvspace{0.2cm}
\begin{tcolorbox}[width=\textwidth,
                  enhanced,
                  boxsep=2pt,
                  left=1pt,
                  right=1pt,
                  top=4pt,
                  boxrule=1pt,
                  arc=0pt,
                  colback=white,
                  colframe=black,
                  unbreakable,
                  float=t
                  ]
}{
\end{tcolorbox}
}

\newenvironment{graytbox}{
\par\addvspace{0.1cm}
\begin{tcolorbox}[width=\textwidth,
                  enhanced,
                  frame hidden,
                  boxsep=5pt,
                  left=1pt,
                  right=1pt,
                  top=2pt,
                  bottom=2pt,
                  boxrule=1pt,
                  arc=0pt,
                  colback=mylightgray,
                  colframe=black,
                  breakable
                  ]
}{
\end{tcolorbox}
}

\newcommand{\tboxhrule}[0]{\vspace{0.1cm} \hrule \vspace{0.2cm}}

\newenvironment{titledtbox}[1]{\begin{mytbox}#1 \tboxhrule}{\end{mytbox}}
\newenvironment{titledtboxh}[1]{\begin{mytboxh}#1 \tboxhrule}{\end{mytboxh}}

\newenvironment{tboxalg}[2][]{\refstepcounter{myalgctr}\begin{titledtbox}{\textbf{Algorithm \themyalgctr}#1\textbf{.} #2}}{\end{titledtbox}}

\newenvironment{tboxalgh}[2][]{\refstepcounter{myalgctr}\begin{titledtboxh}{\textbf{Algorithm \themyalgctr}#1\textbf{.} #2}}{\end{titledtboxh}}

\newcounter{programctr}

\newcolumntype{?}{!{\vrule width 1.5pt}}

\DeclareMathOperator*{\Exp}{{\normalfont \textbf{E}}}
\DeclareMathOperator*{\Prob}{{\normalfont \textbf{Pr}}}

\newcommand{\E}[0]{{\normalfont \textbf{E}}}
\newcommand{\Ex}[0]{\Exp}

\renewcommand{\Pr}[0]{\Prob}

\begin{document}
\maketitle

\pagenumbering{roman}

\begin{abstract}

We study the maximum matching problem in the {\em random-order} semi-streaming setting. In this problem, the edges of an arbitrary $n$-vertex graph $G=(V, E)$ arrive in a stream one by one and in a random order. 
The goal is to have a single pass over the stream, use $O(n \cdot \polylog{(n)})$ space, and output a large matching of $G$.

\medskip

We prove that for an absolute constant $\eps_0 > 0$, one can find a $(2/3 + \eps_0)$-approximate maximum matching of $G$ using $O(n \log n)$ space with high probability. This breaks the natural boundary of $2/3$ for this problem prevalent in the prior work 
and resolves an open problem of Bernstein~[ICALP'20] on whether a $(2/3 + \Omega(1))$-approximation is achievable.
\end{abstract}

\clearpage

\setcounter{tocdepth}{3}
\tableofcontents

\clearpage

\pagenumbering{arabic}
\setcounter{page}{1}


\section{Introduction}\label{sec:intro}

A matching in a graph $G=(V,E)$ is any collection of vertex-disjoint edges and in the maximum matching problem, we are interested in finding a matching of largest size in $G$. 
This problem has been a cornerstone of algorithmic research and its study has led to numerous breakthrough results in theoretical computer science. In this paper, we study the maximum matching
problem in the \emph{semi-streaming} model of computation~\cite{FeigenbaumKMSZ05} defined as follows. 

\begin{definition}\label{def:stream}
	Given a graph $G=(V,E)$ with $n$ vertices $V = \set{1,\ldots,n}$ and $m$ edges in $E$ presented in a stream $S=\langle e_1,\ldots,e_m \rangle$, a semi-streaming algorithm makes a single pass 
	over the stream of edges $S$ and uses $O(n \cdot \polylog{(n)})$ space, measured in words of size $\Theta(\log{n})$ bits, and at the end outputs an approximate maximum matching of $G$. 
\end{definition}

The greedy algorithm for maximal matching gives a simple $\sfrac12$-approximation algorithm to this problem in $O(n)$ space. 
When the stream of edges is adversarially ordered, this is simply the best result known for this problem, while it is also known that  a better than $\frac1{1+\ln{2}} \sim 0.59$-approximation is 
not possible~\cite{Kapralov21} (see also~\cite{Kapralov13,GoelKK12}). Closing the gap between these upper and lower bounds is among the most longstanding open problems in the graph streaming literature. 

Going beyond this ``doubly worst case'' scenario, namely, an adversarially-chosen graph and an adversarially-ordered stream, there has been an extensive interest in recent years in studying this problem on \textbf{random order streams}. 
This line of work was pioneered in~\cite{KonradMM12} who showed that the $\sfrac12$-approximation of greedy  can be broken in this case and obtained an algorithm with approximation ratio $(\sfrac12+0.003)$ for this problem. 
Since~\cite{KonradMM12}, there has been two main lines of attack on this problem. Firstly, \cite{Konrad18,GamlathKMS19,FarhadiHMRR20} followed up on the approach of~\cite{KonradMM12} and improved the approximation ratio all the way to $6/11$~\cite{FarhadiHMRR20}. In parallel, \cite{AssadiBBMS19} built on the sparsification approach of~\cite{BernsteinS15,BernsteinS16} in dynamic graphs to achieve an  (almost) $\sfrac23$-approximation but at the cost of $\Ot(n^{1.5})$ space, which is no longer semi-streaming. A beautiful work of~\cite{Bernstein20} then obtained a semi-streaming (almost) $\sfrac23$-approximation by showing how a generalization of the sparsification approach in \cite{AssadiBBMS19} can be found in $\widetilde{O}(n)$ space.

The $\sfrac23$-approximation ratio of the algorithm of~\cite{Bernstein20} is the best possible among all prior techniques for this problem:  the first line of attack in~\cite{KonradMM12,Konrad18,GamlathKMS19,FarhadiHMRR20} is 
based on finding length-$3$ augmenting paths and even finding \emph{all} these paths does not lead to a better-than-$\sfrac23$-approximation\footnote{The work of~\cite{FarhadiHMRR20} also considers  
length-$5$ augmenting paths. However, these paths are used \emph{instead of} length-$3$ paths ``missed'' by the algorithm \emph{not in addition to}  length-$3$ paths and thus the same shortcoming  persists.}. The second line
in~\cite{AssadiBBMS19,Bernstein20} is based on finding an  \emph{edge-degree constrained subgraph} (EDCS) which hits the same exact barrier as there are graphs whose EDCS does not provide a better than $\sfrac23$-approximation (see~\cite{BernsteinS15}). Finally, even for an algorithmically easier variant of this problem, the one-way communication problem, which roughly corresponds to only measuring the space of the algorithm when crossing the midpoint of the stream, the best known approximation ratio is still $\sfrac23$ which is known to be tight
for adversarial orders/partitions~\cite{GoelKK12}. 

Given this state-of-affairs, the $\sfrac23$-approximation ratio for random-order streaming matching has emerged as natural barrier \cite{Konrad18,Bernstein20}. 
In particular,~\cite{Bernstein20} posed obtaining a $(\sfrac23+\Omega(1))$-approximation to this problem as an important open question. We resolve this question in the affirmative in our work. 

\subsection{Our Contributions}

Our main result is a semi-streaming algorithm for maximum matching in random-order streams with approximation ratio strictly-better-than-$\sfrac23$. 

\begin{graytbox}
	\begin{theorem}[Main Result]\label{thm:main}
	Let $G$ be an $n$-vertex graph whose edges arrive in a random-order stream. For an absolute constant $\eps_0 > 0$, there is a single-pass streaming algorithm that obtains a $(\frac{2}{3} + \eps_0)$-approximate maximum matching of $G$ using $O(n\log n)$ space with high probability.
\end{theorem}
\end{graytbox}

\Cref{thm:main} breaks the $\sfrac23$-barrier of all prior work in~\cite{KonradMM12,Konrad18,GamlathKMS19,AssadiBBMS19,Bernstein20,FarhadiHMRR20}. 
Moreover, even though the improvement over $\sfrac23$ is  minuscule in this theorem (while we did not optimize for constants, the bound on $\eps_0$ is only $\sim 10^{-14}$ at this point), 
it still proves that (\sfrac23)-approximation is not the ``right'' answer to this problem. This is in contrast to some other  problems of similar flavor such as one-way communication complexity 
of matching (on adversarial partitions)~\cite{GoelKK12,AssadiB19} or the fault-tolerant matching problem~\cite{AssadiB19} which are both solved using similar techniques (see the unifying framework of~\cite{AssadiB19} based on EDCS) 
and for both $\sfrac23$-approximation
is provably best possible. 

\paragraph{Beyond $(\sfrac23)$-approximation.} Breaking this $\sfrac23$-barrier naturally raises the question on what is the right bound on the approximation ratio of random-order streaming matching. In particular, is $(1-\eps)$-approximation possible? 
We make progress toward settling this question  by showing that no ``truly'' space-efficient algorithm exists for this latter problem: there is provably no semi-streaming matching algorithm even on bipartite graphs that can achieve a $(1-\eps)$-approximation in $O(\exp((1/\epsilon)^{0.99}) \cdot n \cdot \polylog{(n)})$ space; in other words, if one hopes for achieving a $(1-\eps)$-approximation, an exponential dependence on $(1/\eps)$ in the space is unavoidable (see~\Cref{cor:lower-stream}).

As the main focus of our work is on the algorithm in~\Cref{thm:main}, we postpone the details and the ideas behind this result to~\Cref{sec:lower}.

\subsection{Overview of Techniques}

\paragraph{Prior work.} As stated earlier, there has been two main lines of attack on the streaming matching problem in  random-order streams. The first approach aims to find a \emph{large} matching of the graph $G$ early on in the stream, and 
then spends the rest of the stream \emph{augmenting} this matching. For instance,~\cite{KonradMM12} showed that in order for the greedy algorithm to fail to find a better-than-$\sfrac12$-approximation, the algorithm should necessarily pick 
many ``wrong'' edges early on in the stream. As such, in instances where greedy is not beating the $\sfrac12$-approximation itself, we already have an almost $\sfrac12$-approximation by the \emph{middle} of the stream, and we can thus 
focus on augmenting this matching in the remainder half to beat $\sfrac12$-approximation. The work of~\cite{Konrad18} then improved this result further by showing that a modified greedy algorithm, when unsuccessful in obtaining a large 
matching itself, finds an almost $\sfrac12$-approximation when only $o(1)$-fraction of the stream has passed (as opposed to middle), which gives us more room for augmentation. Finally,~\cite{FarhadiHMRR20} built on this 
approach and further improved the augmentation phase. 

The second approach to this problem was based on obtaining an EDCS, a subgraph defined by~\cite{BernsteinS15,BernsteinS16} and studied further in~\cite{AssadiB19}, that acts as a ``matching sparsifier''. On a high level, 
an EDCS is a sparse subgraph satisfying the following two constraints: $(i)$ edge-degree of edges in the EDCS cannot be ``high'', while $(ii)$ edge-degree of missing edges cannot be ``low''. These constraints 
ensure that an EDCS always contains an almost $\sfrac23$-approximate matching of the graph and has additional robustness properties~\cite{BernsteinS15,BernsteinS16,AssadiBBMS19,AssadiB19,Bernstein20}. For instance,
~\cite{AssadiBBMS19} proved that union of several EDCS computed on different parts of a random stream, is itself an EDCS for the entire stream. This allowed them to compute an EDCS of the input in $\Ot(n^{1.5})$ space 
and directly obtain their almost $\sfrac23$-approximation. Finally,~\cite{Bernstein20} gave an elegant proof that weakening the requirement of EDCS allows one to still preserve the almost $\sfrac23$-approximation 
but now recover this subgraph in only $O(n\log{n})$ space. More specifically, the algorithm of~\cite{Bernstein20} first finds a subgraph only satisfying property $(i)$ of the EDCS in 
the first $o(1)$ fraction of the stream, and then picks \emph{all} (potentially) necessary edges for satisfying property $(ii)$ in the remainder; the proof then shows that this set of potentially necessary edges 
is of size only $O(n\log{n})$. 

\paragraph{Our work.} Our approach can be seen as a natural combination of these two mostly disjoint lines of work. The first part comes from a better understanding of EDCS. 
We present a rough characterization of when an EDCS cannot beat the $\sfrac23$-approximation, which shows that in these instances, we can effectively ignore the second constraint of EDCS.
As a result, we obtain that the only way for the algorithm of~\cite{Bernstein20} to fail to achieve a better-than-$\sfrac23$-approximation, is if it already picks an almost $\sfrac23$-approximation in
the first $o(1)$ fraction of the stream. Note that this is conceptually similar to the first line of work on random-order streaming matching, but the techniques are entirely disjoint. In particular, our proof is a 
{deterministic} property of EDCS not a {randomized} property of a greedy algorithm on a particular ordering.

We are now in the familiar territory of having a large matching very early on in the stream, and we can spend the remainder of the stream augmenting it. The main difference however is that 
starting from an almost $\sfrac23$-approximation matching, there is essentially no length-$3$ paths for us to augment and we instead need to handle length-$5$ augmenting paths. The key challenge is to find the middle edge of these length-5 augmenting paths. Indeed, we note that the $\sfrac23$-approximation lower bound of \cite{GoelKK12} for {\em adversarial} order streams gives away a $\sfrac23$-approximate matching early on for free, yet it is  provably impossible to augment it in the remainder of the stream using a semi-streaming algorithm. To get around this, we crucially use the random arrival assumption again. Particularly, we regard any length-5 augmenting path whose middle edge arrives after its two endpoint edges as a ``discoverable'' path and then find a constant fraction of such paths. Since the edges arrive in a random order, a constant fraction of length-5 augmenting will be discoverable and thus we are able to beat $\sfrac23$-approximation in our setting.
\newcommand{\phaseo}{\ensuremath{\textnormal{Phase I}}\xspace}
\newcommand{\phaset}{\ensuremath{\textnormal{Phase II}}\xspace}

\newcommand{\BB}{\ensuremath{\mathcal{B}}}

\newcommand{\bd}{\ensuremath{\hat{d}}}

\newcommand{\betap}{\beta_{+}}
\newcommand{\betam}{\beta_{-}}

\newcommand{\Mstar}{M^*}
\newcommand{\Nstar}{\bar{M}}

\newcommand{\MstarU}{\Mstar_{U}}
\newcommand{\MstarbU}{\Mstar_{{\bar{U}}}}

\newcommand{\tG}{\ensuremath{\tilde{G}}}
\newcommand{\tE}{\ensuremath{\tilde{E}}}

\newcommand{\tH}{\ensuremath{\tilde{H}}}
\newcommand{\tU}{\ensuremath{\tilde{U}}}

\section{Notation and Preliminaries}\label{sec:structure}

\paragraph{General notation.} For a graph $G=(V,E)$ and $v \in V$, we use $\deg_G(v)$ to denote the degree of $v$ in $G$ and $N_G(v)$ to denote the neighborset of $v$ (when clear from the context, we may drop the subscript $G$). 
For any edge $e=(u,v) \in E$, we define the \emph{edge-degree} of $e$ in $G$ as $\deg(u)+\deg(v)$. We use $\mu(G)$ to denote the {\em size} (i.e., the number of edges) of the maximum matching in $G$. 

For integer $k \geq 1$ and $p \in [0,1]$, we use $\BB(k, p)$ to denote the \emph{binomial distribution} with parameters $k$ and $p$. That is, $\BB(k, p)$ is the discrete probability distribution of the number of successful experiments out of $k$ experiments each with an independent probability $p$ of success.  


\paragraph{Random-order streams.} We consider the random-order streaming setting where the edges of $G$ arrive one by one in an order chosen uniformly at random from all possible orderings. Let $e_i$ be the $i$-th edge that arrives in the stream. For any two parameters $a, b$ satisfying $1 \leq a < b \leq m$ we use $G[a, b]$ to denote the subgraph of $G$ on vertex-set $V$ and edge-set $\{e_{a}, \ldots, e_{b}\}$. We may also use $G_{< a}$ and $G_{\geq a}$ respectively as shorthands for $G[1, a-1]$ and $G[a, m]$. 

For the input graph $G$ defined by the stream, we can assume w.l.o.g. that $\mu(G) \geq c \log n$ for any desirably large constant $c$. The reason is that any graph can be easily shown to have at most $2 n \cdot \mu(G)$ edges and if $\mu(G) = O(\log n)$ then we can  store the whole input in the memory and report an optimal solution using $O(n \log n)$ space. We further assume throughout the paper that the number of edges $m$ is known by the algorithm in advance. This is a common assumption in the literature and can be removed via standard techniques 
by guessing $m$ in geometrically increasing values at the expense of multiplying the space by an $O(\log n)$ factor.

\subsection{Preliminaries}\label{sec:preliminaries}

\paragraph{Probabilistic tools.} We use the following standard forms of Chernoff bound. 
\begin{proposition}[{Chernoff Bound}; cf.~\cite{AlonS04}]\label{prop:chernoff}
	Suppose $X_1,\ldots,X_t$ are $t$ independent random variables with values in $[0,1]$. Let $X := \sum_{i=1}^{t} X_i$ and assume $\Ex\bracket{X} \leq b$. For any $\delta > 0$ and  $k \geq 1$,
	\begin{align*}
		&\Pr\Paren{\card{X - \Ex\bracket{X}} \geq \delta \cdot b} \leq 2\cdot\exp\Paren{-\frac{\delta^2 \cdot {b}}{3}}\quad \& \quad 
		\Pr\Paren{\card{X - \Ex\bracket{X}} \geq k} \leq 2\cdot\exp\Paren{-\frac{2k^2}{t}}.
	\end{align*} 
\end{proposition}

We also need Lov\'asz Local Lemma (LLL) in our proofs. 

\begin{proposition}[{Lov\'asz Local Lemma}; cf.~\cite{AlonS04}]\label{prop:lll}
	Let $p \in (0,1)$ and $d \geq 1$. Suppose $\event_1,\ldots,\event_t$ are $t$ events such that $\Pr\paren{\event_i} \leq p$ for all $i \in [t]$ and each $\event_i$ is mutually independent of all but (at most) $d$ other events $\event_j$. 
	If $p \cdot (d+1) < 1/e$ then $\Pr\paren{\cap_{i=1}^{n}\overline{\event_i}} > 0$. 
\end{proposition}

\paragraph{Hall's theorem.} We use the following standard extension of the Hall's marriage theorem for characterizing maximum matching size in bipartite graphs.

\begin{fact}[Extended Hall's Theorem; cf.~\cite{Hall35}]\label{prop:halls-marriage}
	Let $G=(L,R,E)$ be a bipartite graph and $\card{L} = \card{R} = n$. Then,
	\[
		\max \Paren{\card{A} - \card{N(A)}} = n - \mu(G), 
	\]
	where $A$ ranges over $L$ or $R$, separately. We refer to such set $A$ as a \textbf{witness set}. 
\end{fact}

\Cref{prop:halls-marriage} follows from Tutte-Berge formula for matching size in general graphs~\cite{Tutte47,Berge62} or a simple extension of the proof of Hall's marriage theorem itself.\footnote{Simply add $n-\mu(G)$ vertices to each side of the graph
and connect them to all the original vertices; then apply original's Hall's theorem for perfect matching to this graph as this graph now has one.}

\paragraph{Alternating and augmenting paths.} Given a matching $M$, an {\em alternating path} $P$ for $M$ is a path whose edges alternatively belong to $M$ and do not belong to $M$. An {\em augmenting path} for $M$ is an alternative path that starts and ends with edges that do not belong to $M$. Given an augmenting path $P$ for $M$, we use notation $M \oplus P := (M \setminus P) \cup (P \setminus M)$ to denote the matching obtained by flipping the containment of edges of $P$ in $M$. Given two matchings $M$ and $M'$, their {\em symmetric difference} $M \Delta M'$ is a graph including only the edges that belong to exactly one of $M$ and $M'$.

\subsection{Bernstein's Algorithm}
We briefly review the  parameters and guarantees of the algorithm of Bernstein~\cite{Bernstein20} that we use in our paper. In the following, we slightly increase the constants in the parameters which is needed
for our results.

\begin{definition}[\textbf{Parameters}]\label{def:betalambdaparameter}
	For some small $\epsilon \in (0, \frac12)$ to be determined later, let 
	\[
	\lambda := \frac{\epsilon}{128}, \qquad \betap := 64 \cdot \lambda^{-2} \log (1/\lambda), \qquad \betam = (1-\lambda) \cdot \betap.
	\]
\end{definition}


A high level overview of the algorithm of~\cite{Bernstein20} is as follows:

\begin{tboxalg}{Bernstein's Algorithm~\cite{Bernstein20}.}\label{alg:bernstein}
The algorithm of \cite{Bernstein20} proceeds in two phases as follows: 
\begin{itemize}[leftmargin=15pt]
\item \phaseo terminates within the first $\epsilon m$ edges of the stream. At the end of \phaseo, the algorithm constructs a subgraph $H \subseteq G_{<\epsilon m}$ such that for all $(u,v) \in H$:
\[
	 \deg_H(u)+\deg_H(v) \leq \betap.
\]
Moreover, let $U$ be the set of \emph{all} edges in $G_{\geq \epsilon m}$ such that
\[
	\deg_H(u)+\deg_H(v) < \betam.
\]

\item In \phaset, the algorithm simply stores $U$ in the memory and at the end of the stream returns a maximum matching of $H \cup U$. 
\end{itemize}
\end{tboxalg}

The following lemma is all we need from \cite{Bernstein20} in our paper.

\begin{lemma}[Lemma~4.1 of \cite{Bernstein20}]\label{lem:phaseI}
	There is a way of constructing the subgraph $H$ of $G_{<\epsilon m}$ such that with probability at least $1-n^{-3}$, $|H \cup U| = O(n \log{(n)} \cdot \poly(1/\epsilon))$.
\end{lemma}


\section{Finding an Almost $(\sfrac{2}{3})$-Approximation Early On} \label{sec:earlyon}

We start by characterizing the tight instances of the algorithm of~\cite{Bernstein20} (\Cref{alg:bernstein}). Roughly speaking, we show that 
 the only way for \Cref{alg:bernstein} to end up with a $(2/3)$-approximation is if in its \phaseo  it computes a subgraph $H$ that already has an almost $(2/3)$-approximate matching. 
 This will then be used by our algorithm in the next section to obtain a strictly better-than-$(2/3)$-approximation by augmenting this already-large matching. 
 
 We start by presenting and proving this result for bipartite graphs which is the main part of the proof; 
 we then extend the result to general graphs (with no considerable loss of parameters for our purpose) using the probabilistic method approach of~\cite{AssadiB19} for the original EDCS.

\subsection{Bipartite Graphs}

In this section we prove the following structural result:

\begin{theorem}\label{lem:edcs-early}
	Let $\lambda \in (0,1/2)$ and $\betam \leq \betap$ be such that $\betap \geq \frac{10}{\lambda}$ and $\betam \geq (1-\lambda) \betap$.  
	Suppose $G = (L,R,E)$ is any bipartite graph and: 
	\begin{enumerate}[label=$(\roman*)$]
		\item $H$ is a subgraph of $G$ where for all $(u,v) \in H$: $\deg_{H}(u) + \deg_{H}(v) \leq \betap$; and
		\item $U$ is the set of all edges $(u,v)$ in $G \setminus H$ such that $\deg_{H}(u) + \deg_H(v) < \betam$. 
	\end{enumerate}
	Then, for any parameter $\delta \in (0,1)$, either: 
	\[
	  \mu(H) \geq (1-4\lambda) \cdot (\frac23-\delta) \cdot \mu(G)	\quad \text{or} \quad \mu(H \cup U) \geq (1-2\lambda) \cdot \paren{\frac23 + \frac{\delta^2}{18}} \cdot \mu(G). 
	\]
\end{theorem}

\noindent
Let us define the following (see~\Cref{fig:structural} for an illustration):
\begin{itemize}
\item Let $\Mstar$ be a maximum matching of $G$ and define $\MstarU := \Mstar \cap U$ and $\MstarbU := \Mstar \setminus U$. 
\item $A$ is Hall's theorem witness set in $H \cup \MstarU$ (as in~\Cref{prop:halls-marriage}) and $B := N_{H \cup \MstarU}(A)$. Without loss of generality we assume $A \subseteq L$ and define $\bar{A} := L \setminus A$ and $\bar{B} := R \setminus B$. 
\end{itemize}
\noindent
We start with the following simple claim that follows easily from~\Cref{prop:halls-marriage}. 

\begin{claim}\label{clm:easy-cor}
	For the witness set $A$: 
	\begin{enumerate}[label=$(\roman*)$]
		\item $\card{\bar{A}} + \card{{B}} \leq \mu(H \cup U)$. 
		\item There is a matching $\Nstar \subseteq \MstarbU$ between $A$ and $\bar{B}$ in $G$ with size $\card{\Nstar} = \mu(G) - \mu(H \cup \MstarU)$. 
	\end{enumerate}
\end{claim}
\begin{proof}
	For part $(i)$, note that $\card{\bar{A}} + \card{{B}} = n-(\card{A}-\card{B}) = n-(n-\mu(H \cup \MstarU)) \leq \mu(H \cup U)$ where the second to last equation is since $A$ is a witness set in $H \cup \MstarU$, and the 
	last equation is because $\MstarU$ is a subset of $U$.
	
	For part $(ii)$, consider the graph consisting of only $\Mstar$. Given that for the set $A$ in this new graph, we have $\card{A}-\card{N_{\Mstar}(A)} \leq n-\mu(G)$ by~\Cref{prop:halls-marriage}, we get that 
	$\card{N_{\Mstar}(A)}-\card{B} \geq \mu(G)-\mu(H \cup \MstarU)$. Moreover, since $\Mstar$ is a matching, these new neighbors of $A$ are only formed via a matching. Finally, as these edges are missing 
	from $H \cup \MstarU$, this matching from $A$ to $\bar{B}$ should entirely belong to $\MstarbU$. 
\end{proof}

\bigskip

\begin{figure}[h!]
\centering
	\begin{tikzpicture}
	\node[rounded rectangle, draw, line width=1pt, minimum width=6cm, minimum height=1cm, blue!75, fill=blue!1] (A) {};
	\node (nA) [left=5pt of A]{$\textcolor{blue}{\mathbf{A}}$};
	\node[rounded rectangle, draw, line width=1pt, minimum width=3cm, minimum height=1cm, blue!75, fill=blue!1] (bA) [right=0.5cm of A]{};
	\node (nbA) [right=5pt of bA]{$\textcolor{blue}{\mathbf{\bar{A}}}$};
	\node[rounded rectangle, draw, line width=1pt, minimum width=3cm, minimum height=1cm, blue!75, fill=blue!1] (B) [below left=1cm and -1.75cm of A]{};
	\node (nB) [left=5pt of B]{$\textcolor{blue}{\mathbf{B}}$};
	\node[rounded rectangle, draw, line width=1pt, minimum width=6cm, minimum height=1cm, blue!75, fill=blue!1] (bB) [right=0.5cm of B] {};
	\node (nbB) [right=5pt of bB]{$\textcolor{blue}{\mathbf{\bar{B}}}$};

	\node[rounded rectangle, draw, line width=0.5pt, minimum width=3cm, minimum height=0.75cm, , red!75, fill=red!10] (SL) [right=-3cm of A]{$\mathbf{S}$};
	\node[rounded rectangle, draw, line width=0.5pt, minimum width=3cm, minimum height=0.75cm, , red!75, fill=red!10] (SR) [left=-3cm of bB]{$\mathbf{S}$};
	
	\draw[line width=1pt, red!75] 
		(SL.south west) to (SR.north west)
		(SL.south) to (SR.north)
		(SL.south east) to (SR.north east);
	
	\node (N) [below right=0.3cm and -1.5cm of SL]{$\textcolor{red}{\mathbf{\Nstar}}$};
	
	\node[rounded rectangle, draw, line width=0.5pt, minimum width=2.5cm, minimum height=0.75cm, , ForestGreen!75, fill=ForestGreen!10] (TL) [right=-2.5cm of bA]{$\mathbf{T}$};
	\node[rounded rectangle, draw, line width=0.5pt, minimum width=2.5cm, minimum height=0.75cm, , ForestGreen!75, fill=ForestGreen!10] (TR) [left=-2.5cm of B]{$\mathbf{T}$};

	\draw[line width=1pt, dashed, ForestGreen!75] 
		(SL) to node[left] {$\textcolor{ForestGreen}{\mathbf{F}}$} (TR)
		(SR) to node[left] {$\textcolor{ForestGreen}{\mathbf{F}}$} (TL);

	
\end{tikzpicture}
	\caption{An illustration of the Hall's witness set and our notation in the proof of~\Cref{lem:edcs-early}. Note that in particular, there are no edges between $A$ and $\bar{B}$ in $H \cup \MstarU$, and the matching $\Nstar$ 
	belongs entirely to $\MstarbU$.}\label{fig:structural} 
\end{figure}
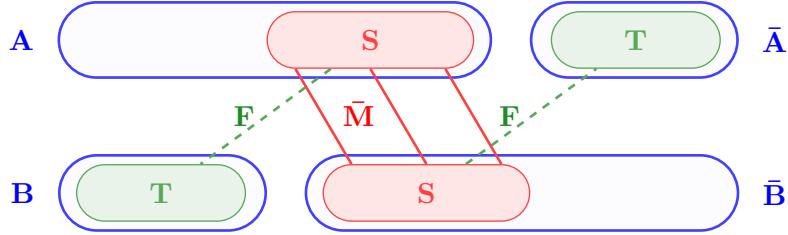

\bigskip

Consider any edge $(u,v) \in \Nstar$ defined in~\Cref{clm:easy-cor}. As $\Nstar \subseteq \MstarbU$, by property $(ii)$ of~\Cref{lem:edcs-early} statement, we have, 
$
	\deg_H(u) + \deg_H(v) \geq \betam. 
$
We arbitrarily remove the edges on $u$ and $v$ until the above inequality becomes tight for every edge (since $\Nstar$ is a matching, this is possible indeed). We let $F$ be the remaining edges. Note that any edge in $F$ 
is incident on exactly one vertex of $\Nstar$ as there are no edges in $H \cup \MstarU$ between the endpoints of $\Nstar$. We record these properties as follows: 
\begin{align}
	\forall (u,v) \in \Nstar: ~	\deg_F(u) + \deg_F(v) = \betam  \quad \text{and} \quad \card{F} = \card{\Nstar} \cdot \betam.\label{eq:N-betam}
\end{align}

In the following, we first give some illustrating examples that highlight the ideas for proving~\Cref{lem:edcs-early}, and then proceed to the formal proof. 

\subsection*{Illustrating Examples and The High Level Idea} 

By~\Cref{clm:easy-cor}, $\mu(H \cup U) \geq \mu(G) - \card{\Nstar}$; thus, if $\Nstar$ is sufficiently smaller than $\mu(G)/3$, we already satisfy the second condition of~\Cref{lem:edcs-early} and we would be done. 
As such, in this informal discussion, we are simply going to assume that $\card{\Nstar} = \mu(G)/3$. Moreover, we define the endpoints of $\Nstar$ as $S$, and their neighborset of $S$ in $H$ as the set $T$. 
See~\Cref{fig:structural} for an illustration. 
Let us now consider two \emph{extreme} cases: 

\paragraph{When degrees of edges in $\Nstar$ are ``highly balanced''.} That is, both endpoints of edges in $\Nstar$, namely, vertices in $S$, have degree $\betam/2$ (recall that by~\Cref{eq:N-betam}, edge-degree of every edge in $\Nstar$ is $\betam$). We claim that in this case, there is a large matching in $H$ already that satisfies condition one of~\Cref{lem:edcs-early}. 

Firstly, note that the degrees of vertices in $T$ needs to be at most $\betap - \betam/2 \leq (1+\lambda)\betap/2$ to satisfy property $(i)$ of~\Cref{lem:edcs-early} for 
edges of $H$ between $S$ and $T$. As such, the subgraph between $S$ and $T$ has degree $\betam/2$ on the $S$-side and degree at most $\betap/2$ on the $T$-side. 
By putting a mass of $\frac{2}{(1+\lambda)\betap}$ on every edge of this subgraph, we can create a feasible fractional matching of value $\card{S} \cdot (\betam/2) \cdot (2/((1+\lambda)\cdot\betap)) \geq (1-\Theta(\lambda))\card{S}$ in this subgraph (and thus $H$).
Considering the integrality gap of the matching polytope in bipartite graphs is one, this means there is a matching of size $(1-\Theta(\lambda))\card{S} = (1-\Theta(\lambda)) \cdot 2\card{\Nstar} = (1-\Theta(\lambda)) \cdot 2\mu(G)/3$ in $H$. Thus, in this case, 
$H$ already has a large matching that satisfies the first condition of~\Cref{lem:edcs-early}. 

It is worth mentioning that the tight $2/3$-approximation example of~\cite{BernsteinS15} for EDCS can be used here to prove that in this case, the subgraph $H \cup U$ may \emph{not} have a matching of size larger than $2\mu(G)/3$, i.e., the 
second condition of~\Cref{lem:edcs-early} may indeed not hold here. 

\paragraph{When degrees of edges in $\Nstar$ are ``mostly unbalanced''.} Let us for our informal discussion assume that for every edge in $\Nstar$ its endpoint in $L$ has degree $\betam/3$ while its endpoint in $R$ has degree $2\betam/3$ (again recall that sum of these degrees should add up to $\betam$ by~\Cref{eq:N-betam}). We claim that in this case, $H \cup U$ has a large matching that satisfies condition two of~\Cref{lem:edcs-early}. 
 
In this case, to satisfy property $(i)$ of~\Cref{lem:edcs-early} for edges of $H$ between $S$ and $T$, we need that vertices in $T \cap L$ should have degree at most $\betap - 2\betam/3 \leq (1+\lambda)\betap/3$. 
Given the bound of $2\betam/3$ on the degrees of vertices in $S \cap R$, we have that,
\[
	\card{T \cap L} \geq (1-\Theta(\lambda)) \cdot 2 \cdot \card{S \cap R}.
\]
 A similar argument also proves that 
 \[
 	\card{T \cap R} \geq (1-\Theta(\lambda)) \cdot \frac12 \cdot \card{S \cap L}.
 \]
 Now note that by~\Cref{clm:easy-cor}, $\card{S \cap R} = \card{S \cap L} = \card{\Nstar} = \mu(G) - \mu(H \cup \MstarU) \geq \mu(G) - \mu(H \cup U)$, while 
 $\card{T \cap L} + \card{T \cap R} = \card{T} \leq \card{\bar{A}} + \card{B} \leq \mu(H \cup U)$. Combining these with the above two bounds, we get that, 
 \[
 \mu(H \cup U) \geq (1-\Theta(\lambda)) \cdot \frac57 \cdot \mu(G). 
 \]
 Thus, in this case, $H \cup U$ has a matching which is a (much) better than $2/3$ approximation. 
 
It is worth mentioning that in this case, the subgraph $H$ may \emph{not} have a matching larger than $3/2 \cdot \card{\Nstar} = \mu(G)/2$, which means the first condition of~\Cref{lem:edcs-early} may indeed not hold here. 

\medskip
The above extreme examples suggest that when edge-degrees of ${\Nstar}$ are more toward being balanced, the subgraph $H$ has a close to $2/3$-approximate matching, 
while when edge-degrees are more unbalanced, the matching of $H \cup U$ is strictly better than $2/3$-approximation. This will be the general strategy underlying our proof of~\Cref{lem:edcs-early} in the 
next subsection. The proof can then be seen more or less as a ``smooth interpolation'' between these two extreme cases. 

\subsection*{The Formal Proof}

In the following lemma, we prove a lower bound on $\mu(H)$. This lemma can then be used as follows: if degree of most edges in $\Nstar$ are ``balanced'', i.e., both endpoints have degree $\approx \betam/2$, then $\mu(H)$ will already be  
of size $2 \cdot \card{\Nstar}$ which will be sufficient for the first condition of~\Cref{lem:edcs-early}.  

\begin{lemma}[\textbf{matching of $H$ is large}]\label{lem:mu(H)-large}
	We have $\mu(H) \geq \frac{\betam}{1+4\lambda} \cdot \sum_{(u,v) \in \Nstar}\frac1{\max\set{\deg_F(u)\, ,\, \deg_F(v)}}$. 
\end{lemma}
\begin{proof}
	For every edge $(u,v) \in \Nstar$, define $F(u,v)$ as  set of edges in $F$ that are incident on $u$ or $v$. We define the following fractional matching $x \in \IR^F$ on edges of $F$: 
	\begin{itemize}
	\item for any edge $e \in F(u,v)$: set $x_e := \frac{1}{1+4\lambda} \cdot \frac{1}{\max\set{\deg_F(u)\, ,\, \deg_F(v)}}$. 
	\end{itemize} 
	Let us now prove that this is indeed a valid fractional matching. For any vertex $w$ matched by $\Nstar$, 
	\[
		x_w:= \sum_{e \ni w} x_e \leq \deg_F(w) \cdot \frac{1}{1+4\lambda} \cdot \frac{1}{\deg_F(w)} < 1,
	\]
	thus satisfying the fractional matching constraint. 
	
	Now fix a vertex $w$ not matched by $\Nstar$. Let $u_1,\ldots,u_{\deg_F(w)}$ denote the neighbors of $w$ in $F$. By definition, all these vertices are matched by $\Nstar$. Let $v_1,\ldots,v_{\deg_F(w)}$ be the matched pairs of 
	these vertices. We need the following simple claim. 
	
	\begin{claim}\label{clm:mu(H)-large}
	For every $i \in [\deg_F(w)]$, 
	$
	\deg_F(w) \leq (1+4\lambda)\cdot \max\set{\deg_F(u_i)\, ,\, \deg_F(v_i)}.
	$ 
	\end{claim}
	\begin{proof}
	We first have the following two equations: 
	\begin{align*}
		\deg_F(w) + \deg_F(u_i) &\leq \betap,   \tag{by the property $(i)$ of~\Cref{lem:edcs-early} statement} \\
		\deg_F(u_i) + \deg_F(v_i) &= \betam. \tag{by~\Cref{eq:N-betam}}
	\end{align*}
	As such, 
	\begin{align*}
		\deg_F(w) - \deg_F(v_i) \leq \betap - \betam \leq 2\lambda \betam \tag{as $\lambda \leq 1/2$, and $\betam \geq (1-\lambda)\betap$}
	\end{align*}
	Noting that $\max\set{\deg_F(u_i)\, ,\, \deg_F(v_i)} \geq \betam/2$ by~\Cref{eq:N-betam}, concludes the proof. 
	\Qed{clm:mu(H)-large}
	
	\end{proof}
	To finalize~\Cref{lem:mu(H)-large}, for any vertex $w$ not matched by $\Nstar$, we have, 
	\[
		x_w := \sum_{e=(w,u_i)} x_e = \sum_{u_i} \frac{1}{1+4\lambda} \cdot \frac{1}{\max\set{\deg_F(u_i)\, ,\, \deg_F(v_i)}} \Leq{\Cref{clm:mu(H)-large}} \sum_{u_i} \frac{1}{\deg_F(w)} = 1,
	\]
	thus satisfying the fractional matching constraint. This implies that $x$ is a valid fractional matching. 
	
	Finally, the value of this fractional matching is: 
	\begin{align*}
		\sum_{e \in F} x_e &= \sum_{(u,v) \in N} \sum_{e \in F(u,v)} x_e = \sum_{(u,v) \in N} \frac{\deg_F(u)+\deg_F(v)}{(1+4\lambda) \cdot {\max\set{\deg_F(u)\, ,\, \deg_F(v)}}} \\
		&=  \frac{\betam}{1+4\lambda} \cdot \sum_{(u,v) \in N}\frac1{\max\set{\deg_F(u)\, ,\, \deg_F(v)}},
	\end{align*}
	where the last equation is by~\Cref{eq:N-betam}. As the integrality gap of matching polytope on bipartite graphs is one, we obtain that the desired lower bound on $\mu(H)$. \Qed{lem:mu(H)-large}
	
\end{proof}

We now prove that if on the other hand most edges of $\Nstar$ are ``unbalanced'', then $\mu(H \cup U)$ should be sufficiently large.
To continue, we need a quick definition. Let $S$ denote the endpoints of the matching $\Nstar$ and $T$ be the neighborset of these vertices in $F$. Recall that by~\Cref{eq:N-betam}, $S$ and $T$ are disjoint (see~\Cref{fig:structural}).

\begin{lemma}[\textbf{matching of $\mu(H \cup U)$ is large}]\label{lem:mu(HU)-large}
	We have $\mu(H \cup U) \geq  \frac{\card{\Nstar}^2 \cdot {\betam}^2}{\card{\Nstar} \cdot \betam \cdot \betap -\sum_{s \in S} (\deg_F(s))^2}$. 
\end{lemma}
\begin{proof}

	Since $F \subseteq H$, by property $(i)$ of~\Cref{lem:edcs-early}, we have that 
	\begin{align}
		\card{F} \cdot \betap \geq \sum_{(u,v) \in F} \deg_F(u)+\deg_F(v) = \sum_{s \in S} (\deg_F(s))^2 + \sum_{t \in T} (\deg_F(t))^2. \label{eq:Zub-large-1}
	\end{align}

	We can lower bound the second term of the RHS as follows. Recall that sum of quadratics is minimized over all-equal terms. As $\sum_{t \in T} \deg_F(t) = \card{F}$, this implies that,  
	\[
		\sum_{t \in T} (\deg_F(t))^2 \geq \sum_{t \in T} (\frac{\card{F}}{\card{T}})^2 = \card{T} \cdot (\frac{\card{F}}{\card{T}})^2 = \frac{\card{F}^2}{\card{T}}.
	\]
	By plugging in this bound in~\Cref{eq:Zub-large-1} and moving the terms around, we have that
	\[
		\card{T} \geq \frac{\card{F}^2}{\card{F}\cdot \betap-\sum_s (\deg_F(s))^2} = \frac{\card{\Nstar}^2 \cdot {\betam}^2}{\card{\Nstar} \cdot \betam \cdot \betap -\sum_s (\deg_F(s))^2}. \tag{as $\card{F}=\card{\Nstar} \cdot \betam$ by~\Cref{eq:N-betam}}
	\]
	Finally, $T \subseteq \bar{A} \cup B$ (as there are no edges between $A$ and $\bar{B}$) and thus by~\Cref{clm:easy-cor}, $\card{T} \leq \mu(H \cup U)$ which finalizes the proof. \Qed{lem:mu(HU)-large}
	
\end{proof}

~\Cref{lem:mu(HU)-large} can be used as follows: when degree of most edges in $\Nstar$ are ``balanced'', the quantity $\sum_s (\deg_F(s))^2$ will be close to $\card{\Nstar} \cdot (\betam)^2/2$ which implies that 
$\mu(H \cup U)$ will be almost $2 \cdot \card{\Nstar}$; however, when degrees of edges in $\Nstar$ are ``unbalanced'', the quantity $\sum_s (\deg_F(s))^2$ \emph{cannot} decrease all the way to  $\card{\Nstar} \cdot (\betam)^2/2$
and thus we can get a higher lower bound on the value of $\mu(H \cup U)$ which breaks the $(2/3)$-approximation. 

To finalize the proof of~\Cref{lem:edcs-early}, we need the following claim for lower bounding $\sum_{s \in S} (\deg_F(s))^2$ in the RHS of~\Cref{lem:mu(HU)-large}, in the cases where 
RHS of~\Cref{lem:mu(H)-large} is small. 

\begin{claim}\label{clm:degS}
	Suppose $\sum_{(u,v) \in \Nstar}\frac{\betam}{\max\set{\deg_F(u)\, ,\, \deg_F(v)}} = (2-\gamma) \cdot \card{\Nstar}$ for  some $\gamma \in [0,1)$; then $\sum_s (\deg_F(s))^2 \geq  \card{\Nstar} \cdot \paren{\frac{(2+\gamma^2-2\gamma) \cdot {\betam}^2}{4+\gamma^2-4\gamma}}$.
\end{claim}
\begin{proof}
	The intuition behind the proof is that $\sum_s (\deg_F(s))^2$ term is a quadratic sum and is thus minimized in the most ``balanced'' case possible under the given constraints. 
	Formally, we define the following vector of  vertex degrees $d \in \IR^{S}$ (recall that $S$ is the endpoints of matching $\Nstar$):
	\begin{itemize}
		\item For any edge $(u,v) \in \Nstar$, let $d_u := \frac{\betam}{2-\gamma}$ and $d_v := \betam - d_u$. 
	\end{itemize} 
	Notice that these vertex degrees satisfy the first constraint of~\Cref{eq:N-betam} and that
	\[
		\sum_{(u,v) \in \Nstar}\frac{\betam}{\max\set{d_u\, ,\, d_v}} = (2-\gamma) \cdot \card{\Nstar},
	\]
	thus satisfying the assumption of the lemma as well. We now prove that these degrees minimize the quadratic sum, namely, 
	\begin{align}
		\sum_{s \in S} (\deg_F(s))^2 \geq \sum_{s \in S} d_s^2. \label{eq:quad-sum} 
	\end{align}
	Suppose there is an edge $(u_1,v_1)$ where $\deg_F(u_1) > d_{u_1}$ and thus $\deg_F(v_1) < d_{v_1}$ (as both pairs satisfy~\Cref{eq:N-betam}).  
	This also implies that there is another edge $(u_2,v_2)$ where $\deg_F(u_2) < d_{u_2}$ and $\deg_F(v_2) > d_{v_2}$ so that the sum of all degrees satisfies the condition of~\Cref{eq:N-betam}. 
	
	Now consider a sufficiently small parameter $\theta_1 \in (0,1)$ and the new ``more balanced'' degrees 
	\begin{align*}
		&\bd_{u_1} := \deg_F(u_1)-\theta_1 \quad , \quad \bd_{v_1} := \deg_F(v_1)+\theta_1, \\
		&\bd_{u_2} := \deg_F(u_2)+\theta_2 \quad , \quad \bd_{v_2} := \deg_F(v_2)-\theta_2,
	\end{align*}
	where $\theta_2$ is defined using the following equation:
	\[
		\frac{1}{\deg_F(u_1)} + \frac{1}{\deg_F(u_2)} = \frac{1}{\deg_F(u_1)-\theta_1} + \frac{1}{\deg_F(u_2)+\theta_2} = \frac{1}{\bd_{u_1}} + \frac{1}{\bd_{u_2}}.
	\]
	Considering $\deg_F(u_1) > \deg_F(u_2)$, we have that $\theta_1 > \theta_2$. Note that these new degrees (assuming we keep the degrees of all other vertices unchanged) satisfy all the constraints as before. 
	We have, 
	\begin{align*}
		\sum_{s \in \set{u_1,v_1,u_2,v_2}} \hspace{-20pt} \deg_F(s)^2 &= (\bd_{u_1}+\theta_1)^2+(\bd_{v_1}-\theta_1)^2 + (\bd_{u_2}-\theta_2)^2+(\bd_{v_2}+\theta_2)^2 \\
		&\geq 2\theta_1 \cdot (\bd_{u_1}-\bd_{v_1}) - 2\theta_2 \cdot (\bd_{u_2}-\bd_{v_2}) + \bd_{u_1}^2 + \bd_{v_1}^2 + \bd_{u_2}^2 + \bd_{v_2}^2 \tag{by ignoring the postive $\theta_1^2,\theta_2^2$ terms} \\
		&> \bd_{u_1}^2 + \bd_{v_1}^2 + \bd_{u_2}^2 + \bd_{v_2}^2  \tag{as $\bd_{u_1}-\bd_{v_1} > \bd_{u_2}-\bd_{v_2}$ and $\theta_1 > \theta_2$}  
	\end{align*}
	Thus, this change reduces the value of $\sum_{s \in S} \deg_F(s)^2$ term as expected. We can now repeatedly continue this until we converge to the degree distribution $\set{d_s}_{s \in S}$ defined earlier. 
	This proves~\Cref{eq:quad-sum}. By plugging in the bounds for $\set{d_s}_{s \in S}$ in the RHS of~\Cref{eq:quad-sum}, we have that,
	\begin{align*}
		\sum_{s \in S} (\deg_F(s))^2 &\geq \sum_{s \in S} (\deg_F(s))^2 = \sum_{(u,v) \in \Nstar} d_u^2 + d_v^2 = \card{\Nstar} \cdot \paren{\frac{{\betam}^2}{(2-\gamma)^2} + (\betam- \frac{{\betam}}{(2-\gamma)})^2} \\
		&= \card{\Nstar} \cdot \paren{\frac{(2+\gamma^2-2\gamma) \cdot {\betam}^2}{4+\gamma^2-4\gamma}},
	\end{align*}
	as desired. \Qed{clm:degS}
	
\end{proof}


\begin{proof}[Proof of~\Cref{lem:edcs-early}]
	Let us pick $\gamma \in [0,1)$ such that $\sum_{(u,v) \in \Nstar}\frac{\betam}{\max\set{\deg_F(u)\, ,\, \deg_F(v)}} = (2-\gamma) \cdot \card{\Nstar}$ (as the max-term is at least $\betam/2$, such a $\gamma$ always exist). 
	By plugging in the bound of~\Cref{clm:degS} in~\Cref{lem:mu(HU)-large}, we have that, 
	\begin{align*}
		\mu(H \cup U) &\geq \frac{\card{\Nstar}^2 \cdot {\betam}^2}{\card{\Nstar} \cdot \betam \cdot \betap -\card{\Nstar} \cdot \paren{\frac{(2+\gamma^2-2\gamma) \cdot {\betam}^2}{4+\gamma^2-4\gamma}}} \\
		&\geq (1-2\lambda) \cdot \card{\Nstar} \cdot \frac{1}{1 -\paren{\frac{(2+\gamma^2-2\gamma)}{4+\gamma^2-4\gamma}}} \tag{as $\betam \geq (1-\lambda)\betap$} \\
		&= (1-2\lambda) \cdot \card{\Nstar} \cdot \frac{4+\gamma^2-4\gamma}{2-2\gamma} = (1-2\lambda) \cdot \card{\Nstar} \cdot (2+\frac{\gamma^2}{2-2\gamma}).
	\end{align*}
	Considering $\card{\Nstar} \geq \mu(G) - \mu(H \cup U)$ by~\Cref{clm:easy-cor}, we obtain that 
	\begin{align*}
		\mu(H \cup U) &\geq (1-2\lambda) \cdot \mu(G) \cdot \paren{\frac23 + \frac{\gamma^2}{18-18\gamma+3\gamma^2}} \geq (1-2\lambda) \cdot \mu(G) \cdot \paren{\frac23 + \frac{\gamma^2}{18}}.
	\end{align*}
	Now if for the parameter $\delta$ in~\Cref{lem:edcs-early}, we already have $\gamma \geq \delta$, we will obtain the second condition. 
	Further, without loss of generality, we can assume that $\card{\Nstar} \geq (\frac{1}{3}-\frac{\delta}{3}) \cdot \mu(G)$ 
	as otherwise $\mu(H \cup \MstarU) \geq (\frac{2}{3}+\delta) \cdot \mu(G)$ by~\Cref{clm:easy-cor} which is stronger than the second condition of~\Cref{lem:edcs-early}.

	Suppose  $\gamma < \delta$ and $\card{\Nstar} \geq (\frac{1}{3}-\frac{\delta}{3}) \cdot \mu(G)$ then. 
	In this case, by the definition of $\gamma$ and~\Cref{lem:mu(H)-large},  
	\begin{align*}
		\mu(H) &\geq \frac{1}{1+4\lambda} \cdot (2-\gamma) \cdot \card{\Nstar} \geq \frac{1}{1+4\lambda} \cdot (2-\delta) \cdot (\frac{1}{3}-\frac{\delta}{3}) \cdot \mu(G) \geq (1-4\lambda) \cdot \paren{\frac{2}{3}-\delta} \cdot \mu(G), 
	\end{align*}
	thus satisfying the first condition. This concludes the proof. \Qed{lem:edcs-early} 
	
\end{proof}

\subsection{General Graphs}

We now extend the results of~\Cref{lem:edcs-early} to general (non-bipartite) graphs following the probabilistic method technique of~\cite{AssadiB19} for the original EDCS. 

\begin{corollary}\label{cor:edcs-early-general}
	Let $\lambda \in (0,1/2)$ and $\betam \leq \betap$ be such that $\betap \geq \frac{64}{\lambda^2}\cdot\log{(1/\lambda)}$ and $\betam \geq (1-\lambda) \betap$.  
	Suppose $G = (V,E)$ is any graph (not necessarily bipartite) and: 
	\begin{enumerate}[label=$(\roman*)$]
		\item $H$ is a subgraph of $G$ where for all $(u,v) \in H$: $\deg_{H}(u) + \deg_{H}(v) \leq \betap$; and
		\item $U$ is the set of all edges $(u,v)$ in $G \setminus H$ such that $\deg_{H}(u) + \deg_H(v) < \betam$. 
	\end{enumerate}
	Then, for any parameter $\delta \in (0,1)$, either: 
	\[
	  \mu(H) \geq (1-8\lambda) \cdot (\frac23-\delta) \cdot \mu(G)	\quad \text{or} \quad \mu(H \cup U) \geq (1-4\lambda) \cdot \paren{\frac23 + \frac{\delta^2}{18}} \cdot \mu(G). 
	\]
\end{corollary}
\begin{proof}
	The proof is based on the probabilistic method and Lov\'asz Local Lemma. Let $\Mstar$ be a maximum matching of $G$. Consider the following randomly chosen bipartite subgraph $\tG=(L,R,\tE)$ of $G$ with respect to $\Mstar$, where
	$L \cup R = V$: 
	\begin{itemize}
		\item For any edge $(u,v) \in \Mstar$, with probability $1/2$, $u$ belongs to $L$ and $v$ belongs to $R$, and with probability $1/2$, the opposite (the choices between different edges of $\Mstar$ are independent). 
		\item For any vertex $w \in V$ not matched by $\Mstar$, we assign $w$ to $L$ or $R$ uniformly at random (again, the choices are independent across vertices). 
		\item The set of edges in $\tE$ are all edges in $E$ with one end point in $L$ and the other one in $R$. 
	\end{itemize}
	Note that by the definition of $\tG$, every edge of $\Mstar$ belongs to $\tG$ as well and thus $\mu(\tG) = \mu(G)$. 
	Define $\tH := H \cap \tG$ and $\tU := U \cap \tG$. We prove that with non-zero probability: 
	\begin{enumerate}[label=$(\roman*)$]
		\item\label{prop1} For all $(u,v) \in \tH$: $\deg_{\tH}(u) + \deg_{\tH}(v) \leq (1+\lambda) \cdot \betap/2$; 
		\item\label{prop2} $\tU$ is the set of all edges $(u,v)$ in $\tG \setminus \tH$ where $\deg_{\tH}(u) + \deg_{\tH}(v) < (1-\lambda) \betam/2$; 
	\end{enumerate}
	
	Before proving these parts, let us mention how they imply~\Cref{cor:edcs-early-general}. Consider the subgraph $\tG$ of $G$ and the sets $\tH$ and $\tU$. 
	Since $\tG$ is bipartite and $\tH$ and $\tU$ satisfy the requirements of~\Cref{lem:edcs-early} for parameters $\tilde{\betap} = (1+\lambda) \cdot \betap/2$, $\tilde{\betam} = (1-\lambda) \betam/2$, and $\tilde{\lambda} = \lambda/2$, we get 
	either
	\[
		 \mu(\tH) \geq (1-8\lambda) \cdot (\frac23-\delta) \cdot \mu(\tG)	\quad \text{or} \quad \mu(\tH \cup \tU) \geq (1-4\lambda) \cdot \paren{\frac23 + \frac{\delta^2}{18}} \cdot \mu(\tG). 
	\] 
	As $\tH \subseteq H$, $\tU \subseteq U$, and $\mu(\tG) = \mu(G)$, we obtain the final result (notice that for this argument, we only need existence of $\tH$ and $\tU$ and not a way of finding them; as such, the 
	non-zero probability guarantee completely suffices for us). 
	
	To prove either property, we need the following auxiliary claim. 
	\begin{claim}\label{clm:degrees}
		With non-zero probability, for every vertex $v \in V$, 
		$
			\card{\deg_{\tH}{(v)} - \deg_H(v)/2} < \frac{\lambda}{4} \cdot \betam.
		$
	\end{claim}
	\begin{proof}
		Fix any vertex $v \in V$ and let $N_H(v):= \set{u_1,\ldots,u_{\deg_H(v)}}$ be the neighbors of $v$ in $H$. Let us assume $v$ is assigned to $L$ in $\tG$ (the other case is symmetric). Hence, degree of $v$ in $\tH$ is 
		exactly equal to the number of vertices in $N_H(v)$ that are chosen in $R$. By construction of $\tG$, 
		\[
		\Ex\bracket{\deg_{\tH}(v)} = \begin{cases} (\deg_H(v)+1)/2 & \quad \text{if $v$ is incident on $\Mstar \cap H$} \\ \deg_H(v)/2 & \quad \text{otherwise} \end{cases}. 
		\] 
		Also, if two vertices $u_i,u_j$ in $N_H(v)$ are matched by $\Mstar$, then
		exactly one of them will be a neighbor to $v$ in $\tH$; otherwise the choices are independent. 
		Thus, by Chernoff bound (\Cref{prop:chernoff}), 
		\begin{align*}
			\Pr\paren{\card{\deg_{\tH}{(v)} - \deg_H(v)/2} \geq \frac{\lambda}{4} \cdot \betam} \leq 2\exp\paren{-\frac{\lambda^2 \cdot \betam^2}{8\betam}} \leq 2\exp\paren{-4\log{\betap}} \leq \frac{2}{\betap^4} \tag{as $\betap \geq 64\lambda^{-2}\log{(1/\lambda)}$ and $\betam \geq (1-\lambda) \betap$, we have $\betam \geq 32\lambda^{-2}\cdot\log{\betap}$}. 
		\end{align*}
		For every vertex $v \in V$, define: 
		\begin{itemize}
			\item event $\event_v$: the event that $\card{\deg_{\tH}({v}) - d_v/2} \geq \frac\lambda4 \cdot \betam$. 
		\end{itemize}
		The event $\event_v$ depends only on the choice of vertices in $N_H(v)$ and hence can depend on at most $\betap^2$ other events $\event_u$ 
		for vertices $u$ which are neighbors to $N_H(v)$. As such, we can apply Lovasz Local Lemma (\Cref{prop:lll}) to argue that 
		with a non-zero probability, $\cap_{v \in V} \overline{\event_v}$ happens, which concludes the proof. \Qed{clm:degrees}
		
	\end{proof}
	
	In the following, we condition on the non-zero probability event of~\Cref{clm:degrees}. 

	\paragraph{Proof of property \ref{prop1}.} For any edge $(u,v) \in \tH$, we have, 
	\[
		\deg_{\tH}{(u)} + \deg_{\tH}{(v)} \leq \frac{1}{2} \cdot \paren{\deg_{H}{(u)} + \deg_{H}{(v)}} + \frac\lambda2 \cdot \betam \leq \betap/2 + \frac\lambda2 \cdot \betam \leq (1+\lambda)\cdot \betap/2, 
	\]
	where the second to last inequality is because $(u,v) \in H$. As such all edge $(u,v) \in \tH$ have the desired bound on edge-degree. 
	
	\paragraph{Proof of property \ref{prop2}.} For any edge $(u,v) \in \tG \setminus \tH$ with $\deg_{\tH}(u) + \deg_{\tH}(v) < (1-\lambda) \cdot \betam/2$, 
	\[
		\deg_{H}{(u)} + \deg_{H}{(v)} \leq 2 \cdot \paren{\deg_{\tH}{(u)} + \deg_{\tH}{(v)}} + \frac\lambda2 \cdot \betam < (1-\lambda) \cdot \betam + \frac\lambda2 \cdot \betam < \betam.
	\]
	This implies that this edge belongs to $U$ and thus since $\tU := \tG \cap U$, it also belongs to $\tU$. As a result, any edge with ``low'' edge-degree belongs to $U$. 
		
	This concludes the proof. \Qed{cor:edcs-early-general}
	
\end{proof}

\section{An Improved Algorithm via Augmentation}\label{sec:augmentation}

In this section, we show that the maximum matching of the subgraph $H$ constructed in the early part of the stream of \Cref{alg:bernstein} can be augmented well via the remaining edges. Combined with our \Cref{cor:edcs-early-general} of \Cref{sec:earlyon}, we complete in this section the proof of \Cref{thm:main}. Namely, we show that for some parameter $\eps_0 > 0$, there is a single-pass random-order streaming algorithm (formalized as~\Cref{alg:beats2/3}) that obtains a $(\frac{2}{3} + \eps_0)$-approximate maximum matching of $G$ using $O(n\log n)$ space with high probability of $1-1/\poly(n)$.

\subsection{The Algorithm} 

Our starting point is~\Cref{alg:bernstein}. Recall that this algorithm stores two subgraphs $H$ and $U$ of $G$ of size $O(n \log n)$. Subgraph $H$ is constructed early on, after merely observing $\epsilon m$ edges of the stream. In addition to $H$ and $U$, here we store an additional subset of edges that we use to augment a matching of $H$ with. Particularly, let $M_H$ be an arbitrary maximum matching of $H$. Having  matching $M_H$ early on, in our algorithm we augment $M_H$ using the edges that arrive in the rest of the stream (i.e., \phaset) in parallel to storing $U$. The augmenting paths that we find may be of size up to \emph{five}. This is crucial since we may not have enough augmenting paths of length smaller than five to go beyond $(2/3)$-approximation. Now by plugging our bound of \Cref{cor:edcs-early-general}, it can be shown that either $H \cup U$ includes our desired approximation of strictly better that $2/3$, or $M_H$ is almost a $(2/3)$-approximate matching which coupled with the augmenting paths that we find for it in \phaset leads to our better-than-$(2/3)$-approximation.

To find these augmenting paths, we divide the $(1-\epsilon)m$ edges of Phase II into Phase II.A and Phase II.B. To do this, we first draw a random variable $\tau \sim \BB((1-\epsilon)m, \gamma)$. Phase II.A will then proceed on the edges that arrive up to the
 $\tau$-th edge of Phase II and Phase II.B proceeds on the rest of the edges. Drawing random variable $\tau$ (instead of having a fixed threshold) is particularly useful in the analysis: Conditioned on the edges that are to arrive in Phase II (but not their ordering), each edge now belongs to Phase II.A {\em independently} with probability $\gamma$ and to Phase II.B otherwise. Note that with a fixed threshold, we do not get this independence.

\bigskip

\begin{figure}[h!]
  \centering
  \includegraphics[scale=0.9]{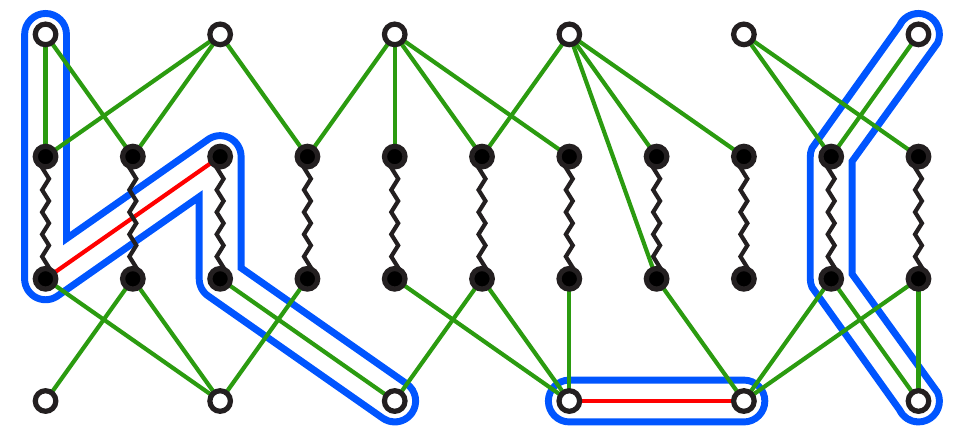}
  \caption{An example of an execution of~\Cref{alg:beats2/3}. Here the black zig-zagged edges are those in matching $M_H$ which is fixed by the end of Phase I and we would like to augment it. The black nodes are those matched by $M_H$ and the white ones are those left unmatched by $M_H$. The edges between white and black nodes (colored green) are the edges in $T$. Each black node has at most two edges in $T$ and the green nodes can have up to $b$. The red edges are those that arrive in Phase II.B. Three augmenting paths of length one, three, and five that are discoverable by the algorithm are also highlighted in the figure.} 
    \label{fig:augmentingpaths}
\end{figure}

\bigskip

For Phase II.A, let us define $G_H$ to be the subgraph of $G$ whose edges arrive in Phase II.A and have exactly one endpoint matched by $M_H$. Note that $G_H$ is bipartite (even though $G$ may not be) with one partition corresponding
 to vertices $V(M_H)$ and another to $V \setminus V(M_H)$. In Phase~II.A, we only consider the edges of $G_H$ and greedily construct a maximal $(2, b)$-matching $T$ of $G_H$ (for some constant $b \geq 2$). It is the vertices in partition $V(M_H)$ of $G_H$ that have maximum degree $2$ in $T$ and those in the other partition can have degree up to $b$. In our analysis, we show that the edges of $T$ can be used as the two endpoint edges of many augmenting paths of length three or five for $M_H$ (see \Cref{fig:augmentingpaths}). 

In Phase II.B, we first let $M \gets M_H$ and upon arrival of each edge $e$, we iteratively augment $M$ via length-up-to-five augmenting paths using the edges in $T \cup \{e\}$ until no such path is left. In our analysis, we use the edges of Phase II.B either as the middle edge of length-five augmenting paths or as the single edge of the length-one augmenting paths the algorithm may find (see \Cref{fig:augmentingpaths}). 

At the end of the stream, we return a maximum matching of $M \cup H \cup U$. The algorithm outlined above is formalized as  Algorithm~\ref{alg:beats2/3}.

\begin{tboxalgh}{A random-order streaming matching algorithm with approximation ratio $> 2/3$.}
	\label{alg:beats2/3}
	\textbf{Parameters:} $\gamma = 2/3$, $b = 500$, and a sufficiently small constant $\epsilon < 0.01$ to be fixed later.
	
	\begin{enumerate}[topsep=5pt, leftmargin=20pt, label={(\arabic*)}]
		\item In Phase I of the algorithm, which consists of the first $\epsilon m$ edges of the stream, we construct a subgraph $H$ of $G$ as in Phase~I of Algorithm~\ref{alg:bernstein}. 
		 At the end of Phase~I, we fix an arbitrary maximum matching $M_H$ of $H$.
		\item In Phase II, which includes all the edges that arrive after Phase II, we store subgraph $U$ using Phase II of Algorithm~\ref{alg:bernstein}. In addition, we store another subset of edges that we use to augment $M_H$. These edges are constructed in two sub-phases Phase II.A and Phase II.B.
		\item Draw random variable $\tau$ from the Binomial distribution $\BB((1 - \epsilon)m, \gamma)$. Note that this can be done in $O(m)$ time and $O(1)$ space as we only need a counter to count the successes.
		\item Phase II.A starts after Phase I and ends upon arrival of the $\tau$'th edge of Phase II.
		\begin{enumerate}
			\item Let $G_{H}(V_H, U_H, E_H)$ be a bipartite subgraph of $G$ where $V_H := V(M_H)$ is the set of vertices matched in $M_H$,  $U_H := V \setminus V(M_H)$ is the set of vertices left unmatched in $M_H$, and $E_H$ is the edges of $G$ between $V_H$ and $U_H$ that arrive in Phase II.A.
			\item We initialize $T \gets \emptyset$ and upon arrival of an edge $e=(u, v)$ of $G_H$ with $u \in U_H$ and $v \in V_H$, if $\deg_T(v) < 2$ and $\deg_T(u) < b$ we add $e$ to $T$. That is, $T$ is a maximal $(2, b)$-matching of $G_H$ which requires $O(n b)$ space to store.
		\end{enumerate}
		\item Phase II.B starts after Phase II.A and continues to the end of the stream:
		\begin{enumerate}
			\item $M \gets M_H$.
			 Upon arrival of each edge $e$ in Phase II.B, we iteratively take an arbitrary augmenting path $P$ for $M$ of \underline{length up to five} using the edges in $M \cup T \cup \{e\}$ and let $M \gets M \oplus P$. We repeat this process until no more augmenting paths of length up to five exist in $M \cup T \cup \{e\}$; we then continue to the next edge of the stream in Phase II.B.
		\end{enumerate}
		\item Finally, we return a maximum matching of $M \cup H \cup U$.
	\end{enumerate}
\end{tboxalgh}

\subsection*{Space Complexity}

We know already from \Cref{lem:phaseI} that $|H \cup U| = O(n \log(n) \cdot \poly(1/\epsilon)) = O(n \log n)$ for constant $\epsilon$ with high probability. In addition, subgraph $T$ that we store in the memory has maximum degree $b = O(1)$ and thus requires $O(n)$ space to store. Other than these, we only store a matching $M$ and augment it only using the edges stored in memory. Hence, overall, the space complexity of the algorithm is $O(n \log n)$ with high probability.

\subsection*{Analysis of Approximation Ratio}

Let $M^\star$ be an arbitrary maximum matching of $G_{\geq \epsilon m}$. Fixing an arbitrary maximum matching of $G$, each of its edges appears in $G_{\geq \epsilon m}$ with probability $(1-\epsilon)$, thus $\E|M^\star| \geq (1-\epsilon)\mu(G)$. Now so long as $\mu(G) \geq 20 \log(n) \epsilon^{-2}$ and $\epsilon < 1/2$ (which we can assume to hold as discussed in \Cref{sec:preliminaries}), we can prove a high probability lower bound on the size of $M^\star$ via a Chernoff bound on negatively associated random variables. See, e.g., \cite[Lemma~2.2]{Bernstein20} for the proof of the following:

\begin{observation}\label{obs:whpMstarlarge}
	If $\mu(G) \geq 20 \log(n) \epsilon^{-2}$ and $\epsilon < 1/2$, then $\Pr[|M^\star| \geq (1-2\epsilon) \mu(G)] \geq 1-n^{-5}$.
\end{observation}

From now on, we condition on $G_{< \epsilon m}$ which fixes subgraph $H$ and matching $M^\star$. We only assume that $G_{< \epsilon m}$ is chosen such that the high probability event of \Cref{obs:whpMstarlarge} holds.

\begin{assumption}\label{ass:Mstarlarge}
	$|M^\star| \geq (1-2\epsilon)\mu(G)$.
\end{assumption}

Other than \Cref{ass:Mstarlarge}, we do not need any other assumption on how $G_{< \epsilon m}$ is chosen for the rest of the analysis of the approximation ratio.\footnote{We note, however, that the randomization in $G_{< \epsilon m}$ is crucial for arguing that the algorithm uses $O(n \log n)$ space. Here, however, we are only analyzing the approximation ratio.} By conditioning on the outcome of \phaseo, the only randomization that will be left, is the order with which the edges of $G_{\geq \epsilon m}$ arrive in the stream. For brevity, we do not explicitly write the conditioning on $G_{< \epsilon m}$ for the rest of the section, but it should be noted that \textbf{all random statements are conditioned on the outcome of Phase I}.

Let $\mc{P}$ be the set of all augmenting paths of $M_H$ in $S := M^\star \Delta M_H$ with length at most five. Note that since we regard $H$ (and thus $M_H$) as given, the set $\mc{P}$ is  deterministic (as it only depends on $M_H$ and $M^\star$ and not on the order of edges in $G_{\geq \epsilon m}$).

\begin{observation}\label{obs:Plarge}
	We have $|\mc{P}| \geq |M^\star| - \frac{4}{3} \cdot \mu(H)$.
\end{observation}
\begin{proof}
	Let $\mc{P}'$ denote the set of augmenting paths of length larger than $5$ in $S$. Note that there must be at least $|M^\star| - |M_H|$ augmenting paths for $M_H$ in $S$, hence $|\mc{P}| + |\mc{P}'| \geq |M^\star| - |M_H|$. Moreover, any augmenting path in $\mc{P}'$ must have at least 3 edges of $M_H$; thus $|\mc{P}'| \leq |M_H|/3$. Combination of the two bounds gives $|\mc{P}| \geq |M^\star| - |M_H| - \frac{1}{3} |M_H| = |M^\star| - \frac{4}{3} |M_H| = |M^\star| - \frac{4}{3} \mu(H)$.
\end{proof}

\newcommand{\GA}{\ensuremath{G_{\text{II.A}}}}
\newcommand{\GB}{\ensuremath{G_{\text{II.B}}}}

We use $\GA$ to denote the subgraph of $G$ that arrives in Phase~II.A and use $\GB$ to denote the subgraph of $G$ that arrives in Phase~II.B.

\begin{definition}\label{def:lucky}
	We say an augmenting path $P \in \mc{P}$ is ``lucky'' under the following conditions:
	\begin{enumerate}
		\item If $P = \langle e_1 \rangle$ then $e_1 \in \GB$.
		\item If $P = \langle e_1, e_2, e_3 \rangle$ then $e_1, e_3 \in \GA$.
		\item If $P = \langle e_1, e_2, e_3, e_4, e_5 \rangle$ then $e_1, e_5 \in \GA$ and $e_3 \in \GB$.
	\end{enumerate}
	We denote the set of lucky augmenting paths in $\mc{P}$ by $\mc{P}_L$.
\end{definition}

Note that the subset $\mc{P}_L$ of $\mc{P}$ is now random since it depends on the order of edges in $G_{\geq \epsilon m}$. \Cref{lem:whpluckyap} below proves that a relatively large fraction of augmenting paths in $\mc{P}$ will turn out to be lucky with high probability. The proof is straightforward and is given in \Cref{sec:expluckyap}.

\begin{lemma}\label{lem:whpluckyap}
	It holds that
	$
	\Pr\Big( |\mc{P}_L| \leq \gamma^2(1-\gamma)|\mc{P}| - \sqrt{15 \mu(G) \ln n} \Big) \leq 2n^{-5}.
	$
\end{lemma}

Next, observe that in Phase II.B of~\Cref{alg:beats2/3} where we iteratively discover augmenting paths, we do not have the whole subgraph $\GA$ and have stored only a subgraph $T$ of $\GA$ in the memory. In addition, when finding augmenting paths we use only the current edge $e$ of $\GB$ in~\Cref{alg:beats2/3}. Therefore, not all lucky paths are actually discoverable by~\Cref{alg:beats2/3}. This motivates our next definition for ``discoverable paths''.

\begin{definition}
	We say an augmenting path $P$ (not necessarily in $\mc{P}$) for $M_H$ is ``discoverable'' if $|P| \leq 5$, all edges of $P$ are in $M_H \cup T \cup \GB$, and $P$ has at most one edge in $\GB$.
\end{definition}

The next lemma proves there are many vertex-disjoint discoverable augmenting paths, by relating them to the number of lucky augmenting paths $|\mc{P}_L|$. We provide the proof in \Cref{sec:proofsizeofQ}.

\begin{lemma}\label{lem:vertex-disjoint-paths-via-T}
	There exists a set $\mc{Q}$ of vertex-disjoint discoverable augmenting paths for $M_H$ with
	$$
	|\mc{Q}| \geq \frac{1}{2b+3}\left(|\mc{P}_L| - \frac{4}{b} \cdot \mu(H) \right).$$ 
\end{lemma}

Observe that $\mc{Q}$ is only a set of vertex-disjoint discoverable augmenting paths. However, since~\Cref{alg:beats2/3} applies augmenting paths greedily and in an arbitrary order, the set of applied augmenting paths may be very different from $\mc{Q}$. The next claim shows that we can nonetheless relate the number of augmenting paths that~\Cref{alg:beats2/3} applies to the size of $\mc{Q}$.

\begin{claim}\label{cl:xgh198273}
	Let $\mc{Q}$ be as in \Cref{lem:vertex-disjoint-paths-via-T}.~\Cref{alg:beats2/3} applies at least $|\mc{Q}|/6$ augmenting paths in \textnormal{Phase II.B}. In other words, $|M| \geq \mu(H) + \frac{1}{6}|\mc{Q}|$.
\end{claim}
\begin{proof}
	Take an augmenting path $P \in \mc{Q}$. Since $P$ is discoverable, there must be a moment during Phase II.B of~\Cref{alg:beats2/3} where all the edges of $P$ are stored in the memory. Note, however, that $P$ is by definition an augmenting path for $M_H$ whereas~\Cref{alg:beats2/3} tries to augment matching $M$ (which is the result of iteratively augmenting $M_H$). The crucial observation, here, is that if $P$ is not an augmenting path for $M$, then at some point one of the augmenting paths that \Cref{alg:beats2/3} has applied on $M$ must have intersected with $P$ (through a vertex). Now, recall that each augmenting paths that~\Cref{alg:beats2/3} applies has length at most five, and thus has at most six vertices. This means that any augmenting path that~\Cref{alg:beats2/3} applies can intersect (and thus ``destroy'') at most six paths in $\mc{Q}$ (since recall $\mc{Q}$ is a collection of vertex-disjoint paths). Hence~\Cref{alg:beats2/3} must apply at least $|\mc{Q}|/6$ augmenting paths on $M$. Since each augmenting path increases the size of $M$ by one and initially $M = M_H$, we have $|M| \geq |M_H| + \frac{1}{6} |\mc{Q}| = \mu(H) + \frac{1}{6} |\mc{Q}|$.
\end{proof}

\begin{lemma}\label{lem:sizeofM}
	There is an absolute constant $\eps'_0 > 0$ such that for any $\epsilon < 0.01$, if $\mu(H) \leq 0.68 \mu(G)$ then with probability $1-1/\poly(n)$, we have $|M| \geq \mu(H) + \eps'_0 \cdot \mu(G)$.
\end{lemma}
\begin{proof}
	We have
	\begin{flalign}\label{eq:bhgc19283}
		|M|
		\stackrel{\text{\Cref{cl:xgh198273}}}{\geq}
		\mu(H) + \frac{1}{6} |\mc{Q}|
		\stackrel{\text{\Cref{lem:vertex-disjoint-paths-via-T}}}{\geq}
		\mu(H) + \frac{|\mc{P}_L| - \frac{4}{b} \mu(H)}{6(2b+3)}.
	\end{flalign}
	On the other hand, by \Cref{lem:whpluckyap} we know that with $1-1/\poly(n)$ probability,
	\begin{flalign*}
		|\mc{P}_L| &> \gamma^2(1-\gamma)|\mc{P}| - \sqrt{15 \mu(G) \ln n} \tag{By \Cref{lem:whpluckyap}}\\
		&=\frac{4}{27} |\mc{P}| - \sqrt{15 \mu(G) \ln n} \tag{Since $\gamma = 2/3$}\\
		&\geq \frac{4}{27} \left(|M^\star| - \frac{4}{3} \mu(H)\right) - \sqrt{15 \mu(G) \ln n}\tag{By \Cref{obs:Plarge}}\\
		&\geq \frac{4}{27} \left((1-2\epsilon)\mu(G) - \frac{4}{3} \mu(H)\right) - \sqrt{15 \mu(G) \ln n}\tag{By \Cref{ass:Mstarlarge}}\\
		&> 0.0108 \mu(G) - \sqrt{15 \mu(G) \ln n}\tag{$\epsilon < 0.01$ and $\mu(H) \leq 0.68\mu(G)$}\\
		&> 0.01 \mu(G). \tag{Since $\mu(G) > c \log n$ for any desirably large constant $c$.}
	\end{flalign*}
	Replacing this high probability lower bound for $|\mc{P}_L|$ into (\ref{eq:bhgc19283}) we get that w.h.p., 
	\begin{flalign*}
		|M| &\geq \mu(H) + \frac{0.01\mu(G) - \frac{4}{b} \mu(H)}{6(2b+3)}\\
		&> \mu(H) + 10^{-7} \mu(G). \tag{Replacing $b = 500$ and noting $\mu(H) \leq 0.68 \mu(G)$.}
	\end{flalign*}
	This completes the proof.
\end{proof}
\noindent
We are now ready to prove that~\Cref{alg:beats2/3}, w.h.p., achieves a better-than-$(2/3)$  approximation.
\begin{lemma}\label{lem:approx}
	For some absolute constant $\eps_0 > 0$ the matching returned by~\Cref{alg:beats2/3} with probability $1-1/\poly(n)$ has size at least $(2/3 + \eps_0) \cdot \mu(G)$.
\end{lemma}
\begin{proof}
	Let $M_O$ be the matching returned by~\Cref{alg:beats2/3} which has size at least as large as maximum of $|M|$ and $\mu(H \cup U)$; we thus get $|M_O| \geq \max\{ |M|, \mu(H \cup U) \}$. Hence, from the lower bound of \Cref{lem:sizeofM} for $|M|$, we get that there is a constant $\eps'_0 > 0$ such that with probability $1-1/\poly(n)$,
	\begin{equation}\label{eq:dlr19283}
		|M_O| \geq 
		\max\Big\{ \mu(H) + \eps'_0 \cdot \mu(G), \,\, \mu(H \cup U) \Big\}.
	\end{equation}
	
	In the next step, we employ \Cref{cor:edcs-early-general} to argue that the lower bound above implies that $|M_O| \geq (2/3 + \Omega(1)) \mu(G)$. In particular, let us consider subgraph $G'$ of $G$ which includes all the edges in $H$ as well as all the edges in $G_{> \epsilon m}$. In other words, the only edges of $G$ that do not belong to $G'$ are those that arrive in Phase~I and are not included in subgraph $H$. One can verify that $H$ and $U$ (constructed in Algorithm~\ref{alg:beats2/3}) satisfy the constraints of \Cref{cor:edcs-early-general} for graph $G'$ (but not necessarily $G$ since the edges in $G - G'$ may have a small edge-degree). \Cref{cor:edcs-early-general} thus implies that for any $\delta \in (0,1)$, either: 
	\[
	  \mu(H) \geq (1-8\lambda) \cdot (\frac23-\delta) \cdot \mu(G')	\quad \text{or} \quad \mu(H \cup U) \geq (1-4\lambda) \cdot \paren{\frac23 + \frac{\delta^2}{18}} \cdot \mu(G'). 
	\]
	Recall that $M^\star$ is the maximum matching of $G_{> \epsilon m}$ which is entirely included in $G'$. Also recall from \Cref{obs:whpMstarlarge} that w.h.p. $|M^\star| \geq (1-2\epsilon) \mu(G)$. Hence, w.h.p., $\mu(G') \geq (1-2\epsilon) \mu(G)$ which combined with $\lambda = \epsilon/128$ (\Cref{def:betalambdaparameter}) simplifies the equation above to the following:
	\begin{equation}\label{eq:xghlll189237}
	  \mu(H) \geq (1-O(\epsilon)) \cdot (\frac23-\delta) \cdot \mu(G)	\quad \text{or} \quad \mu(H \cup U) \geq (1-O(\epsilon)) \cdot \paren{\frac23 + \frac{\delta^2}{18}} \cdot \mu(G). 
	\end{equation}
	
	Plugging (\ref{eq:xghlll189237}) into (\ref{eq:dlr19283}) implies for any $\delta \in (0, 1)$ that
	\begin{flalign*}
		|M_O| &\geq (1-O(\epsilon)) \cdot 
		\min\left\{ \left( \frac{2}{3} - \delta \right) \mu(G) +  \eps'_0 \mu(G), \left(\frac{2}{3} + \frac{\delta^2}{18} \right) \mu(G)\right\}\\
		&\geq (1-O(\epsilon)) \cdot \min\left\{ \left( \frac{2}{3} - \delta + \eps'_0 \right) , \left(\frac{2}{3} + \frac{\delta^2}{18} \right) \right\} \cdot \mu(G).
	\end{flalign*}
	(Note that inequality above takes minimum of the two terms whereas (\ref{eq:dlr19283}) takes maximum. This is because \Cref{cor:edcs-early-general} only guarantees either the lower bound of $\mu(H)$ or that of $\mu(H \cup U)$ and we do not know which one holds for our instance.) 
	
	Now letting $\delta = \eps'_0/2$, we get
	\begin{flalign*}
		|M_O| &\geq (1 - O(\epsilon)) \cdot \min\left\{ \left( \frac{2}{3} + \frac{\eps'_0}{2} \right), \left(\frac{2}{3} + \frac{(\eps'_0 / 2)^2}{18} \right) \right\} \cdot \mu(G) \geq (1-O(\epsilon)) \left(\frac{2}{3} + \frac{(\eps'_0 / 2)^2}{18} \right) \mu(G).
	\end{flalign*}
	Finally, noting that $\epsilon$ can be made arbitrarily small (without affecting $\eps'_0$), combined with the fact that $\eps'_0$ is an absolute positive constant, we get that there must be some $\eps_0 > 0$ such that $|M_O| \geq \left( \frac{2}{3} + \eps_0 \right) \mu(G)$ with probability $1-1/\poly(n)$.
\end{proof}

\Cref{thm:main} now follows immediately from this. 

\subsection{Proof of \Cref{lem:vertex-disjoint-paths-via-T}}\label{sec:proofsizeofQ}

Observe that not all augmenting path $P \in \mc{P}_L$ are discoverable. For example, if $P \in \mc{P}_L$ is of length five, despite its two endpoints $e_1$ and $e_5$ being part of $\GA$ by~\Cref{def:lucky}, it may still be the case that $e_1, e_5 \not\in T$ and thus $e_1, e_5 \not\in M_H \cup T \cup \GB$ implying that $P$ may not be discoverable. To prove \Cref{lem:vertex-disjoint-paths-via-T}, however, we show in this section that for most augmenting paths $P \in \mc{P}_L$, we can modify $P$, particularly, by changing its two endpoint edges (if any and if necessary) and turn $P$ into a discoverable augmenting path $\phi(P)$.

Take an augmenting path $P \in \mc{P}_L$ and recall from definition that $\mc{P}_L \subseteq \mc{P}$ and thus $|P| \in \{1, 3, 5\}$. We define $\phi(P)$ as follows depending on the size of $P$:
\begin{itemize}
	\item $|P| = 1$: In this case, we simply let $\phi(P) \gets P$.
	\item $|P| = 3$: Let $\langle e_1, e_2, e_3 \rangle$ be the edges in $P$ and note that $e_2 \in M_H$ since $P$ is an augmenting path for $M_H$. If edges $e'_1, e'_3 \in T$ exist such that $\langle e_1', e_2, e'_3 \rangle$ forms a length-three augmenting path for $M_H$, we let $\phi(P) \gets \langle e_1', e_2, e'_3 \rangle$. Otherwise, $\phi(P) \gets \emptyset$.
	\item $|P| = 5$: Let $\langle e_1, e_2, e_3, e_4, e_5 \rangle$ be the edges in $P$. Note that $e_2, e_4 \in M_H$ since $P$ is an augmenting path for $M_H$ and $e_3 \in G_{II.B}$ since $P \in \mc{P}_L$. Now if there are edges $e'_1, e'_5 \in T$ such that $\langle e'_1, e_2, e_3, e_4, e'_5 \rangle$ is an augmenting path for $M_H$, we let $\phi(P)$ to denote this path. Otherwise, $\phi(P) \gets \emptyset$.
\end{itemize}

The properties enlisted in \Cref{obs:behu12308} are immediate consequences of construction above:

\begin{observation}\label{obs:behu12308}
	Let $P  \in \mc{P}_L$ and suppose $\phi(P) \not= \emptyset$. It holds that
	\begin{enumerate}
		\item $|\phi(P)| = |P|$.
		\item If $P = \langle e_1, \ldots, e_k \rangle$ and $\phi(P)= \langle e'_1, \ldots, e'_k \rangle$ then $e_i = e'_i$ for any $2 \leq i \leq k-1$.
		\item The endpoint vertices of $\phi(P)$ are unmatched in $M_H$ since it is an augmenting path for $M_H$.
		\item If $|\phi(P)| > 1$ then the two endpoint edges of $\phi(P)$ belong to $T$.
		\item If $\phi(P) \not= \emptyset$, then $\phi(P)$ is discoverable.
	\end{enumerate}
\end{observation}

We let $\Phi := \{ \phi(P) \mid P \in \mc{P}_L, \phi(P) \not= \emptyset \}$. Although each element in $\Phi$ is a discoverable augmenting path for $M_H$, it has to be noted that these augmenting paths may not necessarily be vertex-disjoint. In the first part of the proof, we show that a large fraction of paths in $\Phi$ are vertex-disjoint. In the second part, we show that $\Phi$ is itself large. The combination of these two, gives that there is a large number of vertex-disjoint paths in $\Phi$.

\subsubsection*{A Large Fraction of Paths in $\Phi$ are Vertex-Disjoint}

We first need an auxiliary claim:

\begin{claim}\label{cl:intersectionatendpoint}
	Let $P \in \mc{P}_L$ and $P' \in \mc{P}_L$ be such that $P \not= P'$, $\phi(P) \not= \emptyset$, and $\phi(P') \not= \emptyset$. Then:
	\begin{enumerate}
		\item If $\phi(P)$ and $\phi(P')$ intersect at some vertex $v$, then $v$ is an endpoint of both $\phi(P)$ and $\phi(P')$.
		\item If $e \in \phi(P)$ then $e \not\in \phi(P')$.
	\end{enumerate} 
\end{claim}
\begin{proof}
	Note that $P$ and $P'$ are vertex-disjoint since both belong to $\mc{P}_L \subseteq \mc{P}$. By \Cref{obs:behu12308} part~2, only the endpoint edges of $\phi(P)$ and $\phi(P')$ may differ from $P$ and $P'$ respectively. Combination of these two observations implies that any vertex $v$ that belongs to both of $\phi(P)$ and $\phi(P')$ must be an endpoint of at least one of the two paths. Now using \Cref{obs:behu12308} part~3, we get that $v$ cannot be an intermediate vertex of one path and an endpoint of another since an intermediate vertex must be matched in $M_H$ (as both $\phi(P)$ and $\phi(P')$ are augmenting paths for $M_H$). Hence, $v$ must be an endpoint of both $\phi(P)$ and $\phi(P')$.
	
	To prove the second part, we know from the first part that if $e$ belongs to both $\phi(P)$ and $\phi(P')$, then both of the endpoints of $e$ must be endpoints of paths $\phi(P)$ and $\phi(P')$. This means that we should have $|\phi(P)| = |\phi(P')| = 1$ and $P = P'$ contradicting $P \not= P'$.
\end{proof}

The next claim is the formal statement that a large fraction of paths in $\Phi$ are vertex-disjoint.

\begin{claim}\label{cl:Q>Phi/b}
	There is a subset $\mc{Q} \subseteq \Phi$ such that all the augmenting paths in $\mc{Q}$ are vertex-disjoint and $|\mc{Q}| \geq \frac{1}{2b+3} |\Phi|$ where we recall $b$ is the parameter of~\Cref{alg:beats2/3}.
\end{claim}
\begin{proof}
	We greedily construct $\mc{Q} \subseteq \Phi$ by iterating over the augmenting paths in $\Phi$ in an arbitrary order and including in $\mc{Q}$ any encountered augmenting path $\phi \in \Phi$ which does not intersect with augmenting paths already added to $\mc{Q}$. 
	
	Take an augmenting path $\phi(P) \in \Phi$. We know from \Cref{cl:intersectionatendpoint} part~1, that any other path $\phi(P') \in \Phi$ that intersects $\phi(P)$ must do so at an endpoint vertex of $\phi(P)$. Furthermore, by \Cref{cl:intersectionatendpoint} part~2, $\phi(P')$ and $\phi(P'')$ for $P' \not= P''$ cannot be connected to an endpoint of $\phi(P)$ via the same edge. Hence, any $\phi(P')$ intersecting $\phi(P)$ must do so via a unique edge to an endpoint of $P$. Since the two endpoint edges of any path $\phi(P')$ of size larger than one belong to $T$ by \Cref{obs:behu12308} part~4, and that the maximum degree of $T$ is $b$, there are at most $2b$ such paths intersecting $\phi(P)$. Moreover, at most one path $\phi(P')$ of length one can intersect each endpoint of $\phi(P)$ since $\phi(P') = P'$ for length-one paths and thus all of them are vertex-disjoint. Therefore, overall, $\phi(P)$ intersects at most $2b + 2$ other paths $\phi(P')$.
	
	Now every time that we add a path $\phi(P)$ to $\mc{Q}$, let us remove the remaining paths in $\Phi$ that intersect $\phi(P)$. By our discussion above, every time we add a path to $\mc{Q}$, we remove at most $2b + 2$ other paths from $\Phi$. Hence $|\mc{Q}| \geq \frac{1}{2b + 3} |\Phi|$.
\end{proof}

\subsubsection*{The Set $\Phi$ is Large}

The main statement that $\Phi$ is large is formally given as \Cref{cl:Phi-large}. Before proving it, we need two auxiliary \Cref{cl:xgc139827,cl:tcaalrrc123}.

\begin{claim}\label{cl:xgc139827}
	Let $P = \langle e_1, \ldots, e_k \rangle$ be an augmenting path of length three or five in $\mc{P}_L$. Let us denote the endpoints of $e_1$ and $e_k$ respectively by $(u_1, v_1)$ and $(v_k, u_k)$ where $v_1$ is the vertex connected to $e_2$ and $v_k$ is the vertex connected to $e_{k-1}$. If it holds that 
	\begin{equation}\label{eq:hgrc18923}
	(e_1 \in T \text{ or } \deg_T(v_1) \geq 2) \text{ and } (e_k \in T \text{ or } \deg_T(v_k) \geq 2),
	\end{equation}
	then $\phi(P) \not= \emptyset$.
\end{claim}
\begin{proof}
	It suffices from our construction of $\phi(P)$ to show there are edges $e'_1, e'_k \in T$ such that $\langle e'_1, e_2, \ldots, e_{k-1}, e'_k \rangle$ is an augmenting path for $M_H$. We let $e'_1 \gets e_1$ if $e_1 \in T$ and similarly let $e'_k \gets e_k$ if $e_k \in T$. If $e_1 \not\in T$ but still (\ref{eq:hgrc18923}) holds, then $\deg_T(v_1) \geq 2$. Moreover, by construction of $T$ in~\Cref{alg:beats2/3}, these two edges of $v_1$ are in $U_H$, i.e., the vertices left unmatched by $M_H$. Note that none of these two edges of $v_1$ are connected to the intermediate vertices of $P$ since $P$ is an augmenting-path for $M_H$ and hence all of its intermediate vertices are matched by $M_H$ (and so do not belong to $U_H$). However, it could be that one of these edges is connected to the other endpoint of the augmenting path if the graph is non-bipartite. But this can happen for at most one of the edges of $v_1$ since there are no parallel edges in the graph, which leaves the other edge as a valid option for $e'_1$. In a similar way, if $e_k \not\in T$, we get $\deg_T(v_k) \geq 2$ under (\ref{eq:hgrc18923}) and can pick one of these two edges of $v_k$ to be $e'_k$ such that $\langle e'_1, e_2, \ldots, e_{k-1}, e'_k \rangle$ forms an augmenting path for $M_H$. This completes the proof of the claim that condition (\ref{eq:hgrc18923}) suffices to get $\phi(P) \not= \emptyset$.
\end{proof}

\begin{claim}\label{cl:tcaalrrc123}
	Let $P \in \mc{P}_L$, $e_1 = (u_1, v_1)$, and $e_k = (v_k, u_k)$ be as in \Cref{cl:xgc139827}. Suppose that condition (\ref{eq:hgrc18923}) does not hold for $P$. Then $\deg_T(u_1) \geq b$ or $\deg_T(u_k) \geq b$.
\end{claim}
\begin{proof}
	We first argue that both $e_1$ and $e_k$ are part of graph $G_H$ of Phase II.A of~\Cref{alg:beats2/3}. Toward this, note that since $P \in \mc{P}_L$, we get from~\Cref{def:lucky} that $e_1, e_k \in \GA$. Moreover, since $P$ is by definition an augmenting path for $M_H$, its endpoints $u_1, u_k$ must be unmatched in $M_H$ (implying $u_1, u_k \in U_H$) and vertices $v_1, v_k$ which are intermediate vertices of $P$ must be matched in $M_H$ (implying $v_1, v_k \in V_H$). Hence, both $e_1$ and $e_k$ must belong to $G_H$ (refer to~\Cref{alg:beats2/3}). 
	
	Now let us suppose that (\ref{eq:hgrc18923}) is false since its first clause is false. That is, $(e_1 \not\in T \text{ and } \deg_T(v_1) < 2)$. In this case, knowing that $e_1 \in G_H$, the fact that~\Cref{alg:beats2/3} does not add $e_1$ to $T$ upon processing $e_1$ implies that either $\deg_T(v_1) \geq 2$ or $\deg_T(u_1) \geq b$ (see description of Algorithm~\ref{alg:beats2/3}). The former cannot hold or otherwise the first clause of (\ref{eq:hgrc18923}) would not be false. Hence it should be the case that $\deg_T(u_1) \geq b$. The same argument implies that if (\ref{eq:hgrc18923}) is false for its second clause, then $\deg_T(u_k) \geq b$. The proof is thus complete.
\end{proof}

\begin{claim}\label{cl:Phi-large}
	$|\Phi| \geq |\mc{P}_L| - \frac{4}{b} \cdot \mu(H)$.
\end{claim}
\begin{proof}
	Let $\mc{X} := \{ P \in \mc{P}_L \mid \phi(P) = \emptyset\}$. By definition, $\Phi = \mc{P}_L \setminus \mc{X}$, thus 
	\begin{equation}\label{eq:xdf9981127}
		|\Phi| = |\mc{P}_L| - |\mc{X}|.
	\end{equation}
	It, therefore, suffices to upper bound the size of $\mc{X}$. We do so by double counting the number of edges in $T$.
	
	 Recall that for any $P \in \mc{P}_L$, $|P| \in \{1, 3, 5\}$ by definition of $\mc{P}_L$. Moreover, if $|P| = 1$, then by construction $\phi(P) = P \not= \emptyset$ and thus $P \not\in \mc{X}$. Hence for any $P \in \mc{X}$ it holds that $|P| \in \{3, 5\}$. Now, by \Cref{cl:xgc139827}, condition (\ref{eq:hgrc18923}) should not hold for any $P \in \mc{X}$. This further implies from \Cref{cl:tcaalrrc123} that at least one of the endpoints of each $P \in \mc{X}$ must have degree at least $b$ edges in $T$. Since $\mc{X} \subseteq \mc{P}_L$ and all augmenting paths in $\mc{P}_L$ are vertex disjoint, this means that the endpoints of paths in $\mc{X}$ collectively have at least $|\mc{X}| b$ edges in $T$. Moreover, all of these vertices must be on the $U_H = V \setminus V(M_H)$ partition of graph $G_H$ since each $P \in \mc{X} \subseteq \mc{P}_L$ is an augmenting path for $M_H$ by definition of $\mc{P}_L$. Now we give an alternative way of counting the edges in $T$. Note that any vertex in partition $V_H = V(M_H)$ of $G_H$, has at most 2 edges in $T$ by construction of $T$ in Algorithm~\ref{alg:beats2/3}. Hence, the number of edges in $T$ can be upper bounded by $2 \cdot |V(M_H)| = 2 \cdot 2|M_H| = 4|M_H|$. As such, we get $|\mc{X}| b \leq 4 |M_H|$ and thus $|\mc{X}| \leq 4|M_H|/b$. Plugging this upper bound for $|\mc{X}|$ into (\ref{eq:xdf9981127}) and noting that $|M_H| = \mu(H)$ completes the proof. 
\end{proof}

We are finally ready to formally prove \Cref{lem:vertex-disjoint-paths-via-T}:

\begin{proof}[Proof of \Cref{lem:vertex-disjoint-paths-via-T}]
	Let $\mc{Q} \subseteq \Phi$ be as in \Cref{cl:Q>Phi/b}. All the paths in $\mc{Q}$ are vertex-disjoint. Also:
	$$
		|\mc{Q}|
		\stackrel{\text{\Cref{cl:Q>Phi/b}}}{\geq}
		\frac{|\Phi|}{2b+3}
		\stackrel{\text{\Cref{cl:Phi-large}}}{\geq}
		\frac{1}{2b+3}\left(|\mc{P}_L| - \frac{4}{b} \mu(H)\right).
	$$
	The proof of \Cref{lem:vertex-disjoint-paths-via-T} is thus complete.
\end{proof}

\subsection{Proof of \Cref{lem:whpluckyap}}\label{sec:expluckyap}

We first lower bound $\E|\mc{P}_L|$ and then prove \Cref{lem:whpluckyap} via a concentration bound.

\begin{claim}\label{cl:explucky}
	$\E|\mc{P}_L| \geq \gamma^2 (1-\gamma) |\mc{P}|$.
\end{claim}
\begin{proof}
	Recall again that we regard $\mc{P}$ as fixed as we have conditioned on the outcome of Phase~I. Now whether or not an augmenting path $P \in \mc{P}$ turns out to be lucky depends on the arrival ordering of the edges in $G_{\geq \epsilon m}$. We first show that for any $P \in \mc{P}$, 
	\begin{equation}\label{eq:hgl192387}
		\Pr[P \in \mc{P}_L] \geq \gamma^2(1-\gamma).
	\end{equation}
	(Where, recall, we hide the condition on Phase~I for brevity in our probabilistic statements.) 
	
	The key insight is to note that once we condition on $G_{< \epsilon m}$, an edge $e$ that is to arrive in Phase~II  belongs to $\GA$ independently (than other edges of Phase~II) with probability $\gamma$ and belongs to $\GB$ otherwise (i.e., with probability $(1-\gamma)$). As already discussed at the start of \Cref{sec:augmentation}, this follows from the fact that we do not fix the size of Phase~II.A in Algorithm~\ref{alg:beats2/3} but rather choose it from distribution $B((1-\epsilon)m, \gamma)$. Having this independence, we can prove (\ref{eq:hgl192387}) as follows:
	
	\smparagraph{Proof of Inequality $(\ref{eq:hgl192387})$.} Take an augmenting path $P \in \mc{P}$. Since $\mc{P}$ includes augmenting paths of length up to five, $|P| \in \{1, 3, 5\}$. We prove (\ref{eq:hgl192387}) for all three cases one by one.
	
	First, consider the case where $P$ is of length five and let $P = \langle e_1, e_2, e_3, e_4, e_5 \rangle$. By Definition~\ref{def:lucky}, $P$ is lucky if $e_1, e_5 \in G_{II.A}$ and $e_3 \in G_{II.B}$. The former two events happen with probability $\gamma$ each and the latter happens with probability $(1-\gamma)$. Since the three events, as discussed, are independent, we have 
	\begin{flalign*}
		\Pr[P \in \mc{P}_L] = \gamma^2 (1-\gamma) \qquad\qquad \forall P = \langle e_1, e_2, e_3, e_4, e_5 \rangle \in \mc{P}. 
	\end{flalign*}
	For length-three paths, only the two endpoints should appear in Phase~II.A, hence
	\begin{flalign*}
		\Pr[P \in \mc{P}_L] = \gamma^2 \geq \gamma^2(1-\gamma) \qquad\qquad \forall P = \langle e_1, e_2, e_3 \rangle \in \mc{P}. 
	\end{flalign*}
	For length-one paths, the single edge of the path should appear in Phase II.B, hence:
	\begin{flalign*}
		\Pr[P \in \mc{P}_L] = (1-\gamma) \geq \gamma^2(1-\gamma) \qquad\qquad \forall P = \langle e_1 \rangle \in \mc{P}. 
	\end{flalign*}
	The combination of these cases completes the proof of inequality (\ref{eq:hgl192387}).
	
	\smparagraph{Proof of \Cref{lem:whpluckyap} via inequality $(\ref{eq:hgl192387})$.} By linearity of expectation, we have
	$$
	\E|\mc{P}_L| =
	\sum_{P \in \mc{P}} \Pr[P \in \mc{P}_L] \stackrel{(\ref{eq:hgl192387})}{\geq} \sum_{P \in \mc{P}} \gamma^2(1-\gamma) = \gamma^2(1-\gamma)|\mc{P}|.\qedhere
	$$
\end{proof}

We are now ready to prove \Cref{lem:whpluckyap} via a simple Chernoff bound.

\begin{proof}[Proof of \Cref{lem:whpluckyap}]
	Whether or not an augmenting path $P \in \mc{P}$ turns out to be lucky depends on how its odd edges belong to $G_{II.A}$ and $G_{II.B}$. Since all the augmenting paths in $\mc{P}$ are by definition vertex-disjoint, and since as discussed edges of $G_{\geq \epsilon m}$ belong to $G_{II.A}$ and $G_{II.B}$ independently from each other, we get that the paths in $\mc{P}$ belong to $\mc{P}_L$ independently from each other. By a simple Chernoff bound (\Cref{prop:chernoff}), letting $\delta = \sqrt{\frac{15 \ln n}{\E|\mc{P}_L|}} > 0$, we have
	$$
		\Pr\Big( |\mc{P}_L| \leq (1-\delta)\E|\mc{P}_L| = \E|\mc{P}_L| - \sqrt{15 \E|\mc{P}_L| \ln n} \Big) \leq 2\exp \left( - \frac{\delta^2 \cdot \E|\mc{P}_L|}{3} \right) \leq  2\exp(-5 \ln n) = 2n^{-5}.
	$$
	Since $\E|\mc{P}_L| \geq \gamma^2(1-\gamma)|\mc{P}|$ by \Cref{cl:explucky} and $\E|\mc{P}_L| \leq |\mc{P}| \leq \mu(G)$ this implies that
	$$
		\Pr\Big( |\mc{P}_L| \leq \gamma^2(1-\gamma)|\mc{P}| - \sqrt{15 \mu(G) \ln n} \Big) \leq 2n^{-5}.\qedhere
	$$
\end{proof}


\section{A Lower Bound in Random-Order Streams}\label{sec:lower}

We also prove a lower bound on the approximation ratio of semi-streaming algorithms for bipartite matching on random-order streams. 

\begin{theorem}\label{thm:lower-stream}
	There is a parameter ${\eps_1} = \Theta(\nicefrac{1}{\log{n}})$ such that the following is true. Any streaming algorithm 
	that outputs a $(1-{\eps_1})$-approximation for maximum bipartite matching, in expectation or with constant probability, given one pass over a stream of edges of the input graph in a random order
	requires $n^{1+\Omega(\nicefrac{1}{\log\log{n}})}$ space. 
\end{theorem}

\Cref{thm:lower-stream} provides the first non-trivial lower bound for approximating matching in random-order streams. Prior to our work, only a lower bound of $\Omega(n^2)$ space was known for finding an \emph{exact} maximum matching~\cite{ChakrabartiCM08}. 

A direct corollary of this result is then the following. 

\begin{corollary}\label{cor:lower-stream}
	There is no semi-streaming algorithm  for maximum bipartite matching that for every $\eps > 0$, achieves a $(1-\eps)$-approximation in $O(\exp((1/\eps)^{0.99}) \cdot n \cdot\poly\log{(n)})$ space. 
\end{corollary}

The rest of this section is dedicated to the proof of~\Cref{thm:lower-stream}. The proof of this theorem is based on a new lower bound for (robust) one-way communication complexity of matching that we prove in this paper. 
In the following, we first provide the necessary background and preliminaries and then present the lower bound proof.  

\subsection{Preliminaries for the Lower Bound}
\paragraph{\rs graphs.}
For any graph $G$, a matching $M$ of $G$ is an \emph{induced matching} iff for any
two vertices $u$ and $v$ that are matched in $M$, if $u$ and $v$ are not matched to each other, then there is no edge between $u$ and $v$ in $G$.

\begin{definition}[\rs graph~\cite{RuszaS78}]\label{def:rs-graph-short}
  A graph $G$ is an $(r,t)$-\emph{\rs (RS) graph} iff its edges
  consists of $t$ \emph{pairwise disjoint induced} matchings $M_1,\ldots,M_t$, each of size $r$.
\end{definition}


RS graphs, first introduced by Ruzsa and Szemer\'{e}di~\cite{RuszaS78}, have been extensively studied as they arise naturally in property
testing, PCP constructions, additive combinatorics, streaming lower bounds, etc. (see, e.g.,~\cite{TaoV06,HastadW03,FischerLNRRS02,BirkLM93,AlonMS12,GoelKK12,Alon02,AlonS06,FoxHS15}).  

\paragraph{Communication model.} We work in the standard two-party communication model of Yao~\cite{Yao79} and in particular in the one-way model (see the excellent textbook by Kushilevitz and Nisan~\cite{KushilevitzN97} for the standard definitions). 
The only slight derivation is that we focus on randomly partitioned inputs, wherein the input graph is still chosen adversarially, but every edge in the graph is sent to one of the players chosen independently
and uniformly at random. To our knowledge, this model was first introduced by~\cite{ChakrabartiCM08}. We note that the main resource of interest in this model is the \emph{communication} and in particular the players
are assumed to be {computationally unbounded}. 

In the communication problem we study for bipartite matching, we have an $n$-vertex bipartite graph $G=(L,R,E)$ whose edges are partitioned \emph{randomly} into $\EA$ and $\EB$ given to Alice and Bob, respectively (both players know  $L$ and $R$). 
The goal is to  compute an approximate maximum matching of $G$ by Alice sending a single message to Bob and Bob outputting the solution. The goal is to understand the communication-approximation tradeoff for the problem. 

We note that lower bounds on communication complexity in this model immediately imply space lower bounds for streaming algorithm in random-order streams; see ,e.g.~\cite{ChakrabartiCM08}. 

\subsection{High Level Approach} 

Starting from~\cite{GoelKK12}, all known super-linear-in-$n$ communication lower bounds for \emph{approximating} the maximum
matching problem~\cite{GoelKK12,Kapralov13,Konrad15,AssadiKLY16,AssadiKL17,Kapralov21} are via constructions 
based on \rs (RS) graphs~(\Cref{def:rs-graph-short})\footnote{The only exception is the very recent work of~\cite{DarkK20} in a communication model that allows for edge deletions.}. Our work in this paper is no exception (see~\cite{GoelKK12} for a formal reason why RS graphs are necessary for any lower bound
in the one-way model). However, our key novelty  is a way of making these constructions 
``robust'' so that they can be used even under the random partitioning of the input. 

In more details, the lower bound of~\cite{GoelKK12} gives Alice an RS graph with induced matchings of size $\Theta(n)$ each, and gives Bob  
an ``outside'' matching that matches all vertices of this RS graph, except for one of the induced matchings unknown to Alice; this construction is such that any better-than-$(2/3)$-approximation protocol
needs to include many edges from this special induced matching. However, since Alice is unaware of the identity of this special matching, she is unable to communicate its edges 
with a low communication (much less than the density of the graph).

There are two main challenges in extending this bound to the random partition model: $(i)$ the RS graph edges are now partitioned between both players, and $(ii)$ Alice receives a random subset of edges in the outside matching.  
The first challenge is not that problematic as Alice still receives half the edges of the RS graph in expectation. But the second challenge is   more serious as revealing even a small fraction of 
edges in the outside matching is enough to identify the special induced matching to Alice, hence, enabling her to  focus on sending those edges,  breaking the lower bound. 

In order to circumvent this challenge, we replace edges of this outside matching with a new gadget based on the XOR function. We then show that if Alice misses at least one edge from every one of the XOR-gadgets during the random partitioning 
of the input, the identity of the special induced matching of the RS graph remains hidden to her. By picking these gadgets appropriately, we ensure that this event happens with a  large 
probability and use this in careful information-theoretic argument (instead of the combinatorial arguments in~\cite{GoelKK12}) to conclude the 
proof.

\subsection{The XOR-Gadget} 

We introduce the  following gadget as a key component of our lower bound construction. 

\begin{definition}[\textbf{XOR-Gadget}]
	Let $k > 1$ be an \emph{odd} integer and $(x_1,\ldots,x_k)$ be a $k$-tuple of bits. We define the \textbf{XOR-gadget} of $(x_1,\ldots,x_k)$ as the following graph $\Gxor(x_1,\ldots,x_k)$: 
	\begin{itemize}
		\item There are $2k$ vertices $\set{s,a_1,b_1,a_2,b_2,\ldots,a_{k-1},b_{k-1},t}$ in $\Gxor$. We call $s$ the \emph{start vertex} and $t$ the \emph{final vertex}. 
		\item There are $2k-2$ edges in $\Gxor$ defined as follows using the bits $x_1,\ldots,x_k$: 
		\begin{itemize}
			\item $s$ is connected to $a_1$ if $x_1 = 0$ and otherwise is connected to $b_1$. Similarly, $t$ is connected to $a_{k-1}$ if $x_k = 0$ and otherwise is connected to $b_{k-1}$. 
			\item For any $i \in \set{2,\ldots,k-1}$, $a_{i-1},b_{i-1}$ are connected to $a_i,b_i$, respectively, if $x_i = 0$ and to $b_i,a_i$ otherwise. 
		\end{itemize}
	\end{itemize}
		We use $\Exor(x_i)$ to denote the set of two edges in the gadget that depend on the bit $x_i$. 
	
\end{definition}
\noindent
\Cref{fig:xor}  gives an illustration of XOR-gadgets.

\bigskip

\begin{figure}[h!]
\centering
\subcaptionbox{\footnotesize An example when $k=7$, $(x_1,\ldots,x_7) = (0,0,1,0,0,1,0)$ and so $x_1 \oplus x_2 \oplus \cdots \oplus x_7 = 0$. }%
  [1\linewidth]{

\begin{tikzpicture}

\tikzset{vertex/.style={circle, blue, line width=1pt, fill=blue!0,  draw, inner sep=0pt, minimum width=12pt, minimum height=12pt}}
\tikzset{vset/.style={rectangle, rounded corners, black, fill=black!0, draw, line width=1pt, inner sep=0pt, minimum width=18pt, minimum height=90pt}}

\node[vertex] (s){\scriptsize $s$};
\node [vertex] (A0y0) at ($(s) + (50pt,9pt)$){\scriptsize $a_1$};
\node [vertex] (A0y1) at ($(A0y0)+(0,-18pt)$){\scriptsize $b_1$};

\foreach \x in {1,...,5}{
 	 \pgfmathtruncatemacro{\xx}{\x-1}
	 \pgfmathtruncatemacro{\xxx}{\x+1}
	\node [vertex] (A\x y0) at ($(A\xx y0)+(50pt,0)$){\scriptsize $a_{\xxx}$};
	\node[vertex] (A\x y1) at ($(A\xx y1) + (50pt,0)$){\scriptsize $b_{\xxx}$};
}
\node[vertex] (t) at ($(A5y0) + (50pt,-9pt)$){\scriptsize $t$};

\node (x1) at ($(s)+(25pt,24pt)$) {\scriptsize $x_1=0$};
\node (x2) at ($(x1)+(50pt,0pt)$) {\scriptsize $x_2=0$};
\node (x3) at ($(x2)+(50pt,0pt)$) {\scriptsize $x_3=1$};
\node (x4) at ($(x3)+(50pt,0pt)$) {\scriptsize $x_4=0$};
\node (x5) at ($(x4)+(50pt,0pt)$) {\scriptsize $x_5=0$};
\node (x6) at ($(x5)+(50pt,0pt)$) {\scriptsize $x_6=1$};
\node (x7) at ($(x6)+(50pt,0pt)$) {\scriptsize $x_7=0$};

\draw[line width=1pt, ForestGreen]
	(s.east) -- (A0y0.west)
	(A0y1.east) -- (A1y1.west)
	(A1y0.east) -- (A2y1.west)
	(A2y0.east) -- (A3y0.west)
	(A3y1.east) -- (A4y1.west)
	(A4y0.east) -- (A5y1.west)
	(A5y0.east) -- (t.west);

\draw[line width=0.5pt, ForestGreen, dashed]
	(A0y0.east) -- (A1y0.west)
	(A1y1.east) -- (A2y0.west)
	(A2y1.east) -- (A3y1.west)
	(A3y0.east) -- (A4y0.west)
	(A4y1.east) -- (A5y0.west);
	
	\node[rectangle, black, fill=white, minimum width=400pt, minimum height=1pt, inner sep=1pt](r) at ($(s)+(175pt,-1cm)$){};
\end{tikzpicture}} 
  
  \vspace{0.5cm}
  
\subcaptionbox{\footnotesize An example when $k=7$, $(x_1,\ldots,x_7) = (1,0,1,1,1,1,0)$ and so $x_1 \oplus x_2 \oplus \cdots \oplus x_7 = 1$.}%
  [1\linewidth]{\begin{tikzpicture}
\tikzset{vertex/.style={circle, blue, line width=1pt, fill=blue!0,  draw, inner sep=0pt, minimum width=12pt, minimum height=12pt}}
\tikzset{vset/.style={rectangle, rounded corners, black, fill=black!0, draw, line width=1pt, inner sep=0pt, minimum width=18pt, minimum height=90pt}}

\node[vertex] (s){\scriptsize $s$};
\node [vertex] (A0y0) at ($(s) + (50pt,9pt)$){\scriptsize $a_1$};
\node [vertex] (A0y1) at ($(A0y0)+(0,-18pt)$){\scriptsize $b_1$};
\foreach \x in {1,...,5}{
 	 \pgfmathtruncatemacro{\xx}{\x-1}
	 \pgfmathtruncatemacro{\xxx}{\x+1}
	\node [vertex] (A\x y0) at ($(A\xx y0)+(50pt,0)$){\scriptsize $a_{\xxx}$};
	\node[vertex] (A\x y1) at ($(A\xx y1) + (50pt,0)$){\scriptsize $b_{\xxx}$};
}
\node[vertex] (t) at ($(A5y0) + (50pt,-9pt)$){\scriptsize $t$};

\node (x1) at ($(s)+(25pt,24pt)$) {\scriptsize $x_1=1$};
\node (x2) at ($(x1)+(50pt,0pt)$) {\scriptsize $x_2=0$};
\node (x3) at ($(x2)+(50pt,0pt)$) {\scriptsize $x_3=1$};
\node (x4) at ($(x3)+(50pt,0pt)$) {\scriptsize $x_4=1$};
\node (x5) at ($(x4)+(50pt,0pt)$) {\scriptsize $x_5=1$};
\node (x6) at ($(x5)+(50pt,0pt)$) {\scriptsize $x_6=1$};
\node (x7) at ($(x6)+(50pt,0pt)$) {\scriptsize $x_7=0$};

\draw[line width=1pt, ForestGreen]
	(s.east) -- (A0y1.west)
	(A0y0.east) -- (A1y0.west)
	(A1y1.east) -- (A2y0.west)
	(A2y1.east) -- (A3y0.west)
	(A3y1.east) -- (A4y0.west)
	(A4y1.east) -- (A5y0.west);

\draw[line width=0.5pt, ForestGreen, dashed]
	(A0y1.east) -- (A1y1.west)
	(A1y0.east) -- (A2y1.west)
	(A2y0.east) -- (A3y1.west)
	(A3y0.east) -- (A4y1.west)
	(A4y0.east) -- (A5y1.west)
	(A5y0.east) -- (t.west);

	\node[rectangle, black, fill=white, minimum width=400pt, minimum height=1pt, inner sep=1pt](r) at ($(s)+(175pt,-1cm)$){};
	
\end{tikzpicture}} 
\caption{Solid edges show a maximum matching of the gadget and dashed edges are the remaining edges.}\label{fig:xor}
\end{figure}

\bigskip

The following two lemmas capture the main properties of XOR-gadgets for our purpose. 
The first lemma specifies the connection of XOR-gadgets to the maximum matching problem. 
\begin{lemma}\label{lem:xor-gadget-match}
	Let $k > 1$ be an odd integer and $\Gxor(x_1,\ldots,x_k)$ be some XOR-gadget: 
	\begin{enumerate}[label=$(\roman*)$]
		\item if $x_1 \oplus \cdots \oplus x_k = 0$, then there is a \emph{unique} maximum matching in $\Gxor$ with size $k$ and this matching necessarily matches $t$;
		\item if $x_1 \oplus \cdots \oplus x_k = 1$, then the maximum matching size of $\Gxor$ is $k-1$, and there is a maximum matching in $\Gxor$ that does \emph{not} match $t$. 
	\end{enumerate}
\end{lemma}
\begin{proof}
	For this proof, it helps to refer to~\Cref{fig:xor} as a reference point. 
	
	Consider the unique path $P$ starting from $s$ in $\Gxor$. Each bit $x_i=1$ changes the ``parity'' of the path from an $a$-vertex to a $b$-vertex ($s$ and $t$ are considered $a$-vertices for the purpose of this discussion) and 
	each $x_i=0$ keeps this parity the same. As a result:

		$(i)$ if $x_1 \oplus \cdots \oplus x_k = 0$, then $P$ ends in $t$ and thus $\Gxor$ consists of an odd-length 
		path of length $k$ from $s$ to $t$ and another odd-length path of length $k-2$. The unique maximum matching of such a graph matches both $s$ and $t$ and has size $\ceil{k/2} + \ceil{(k-2)/2} = k$.

		$(ii)$ if $x_1 \oplus \cdots \oplus x_k = 1$, then $P$ does not end in $t$ and thus $\Gxor$ consists of two even-length paths with $k+1$ edges each. Each such path leaves out one of its vertices unmatched 
		necessarily and thus this graph has a maximum matching of size $k-1$ which does not match $t$. 
\end{proof}

This second lemma specifies the ``hiding'' properties these XOR-gadgets. 

\begin{lemma}\label{lem:xor-gadget-hide}
	Let $\Gxor(x_1,\ldots,x_k)$ be a random XOR-gadget obtained by picking each bit $x_i$ independently and uniformly at random. Suppose we partition the edges of $\Gxor$ between Alice and Bob 
	such that for at least one bit $x_i$, Alice has not received neither of the edges in $\Exor(x_i)$. Then, distribution of $x_1 \oplus \cdots \oplus x_k$ is still uniform over $\set{0,1}$ even given Alice's edges. 
\end{lemma}
\begin{proof}
	Follows immediately from the fact that switching any single bit in the XOR function, regardless of any fixed setting of the other bits, switches the value of the function. 
\end{proof}

\subsection{A Hard Distribution of Inputs}\label{sec:hard-dist}

We now describe our distribution of input graphs. For the remainder of the proof, we will use the following parameters (all parameters are defined with respect to some integer $N$): 
\begin{align}
	r := N/3, \qquad t := N^{1+\Omega(\nicefrac{1}{\log\log{N}})}, \qquad  k := 2 \cdot \lceil{\log_{(3/4)}{N}\rceil}+1. \label{eq:lower-par}
\end{align}

Let $\GRS$ be a bipartite $(r,t)$-RS graph with $N$ vertices on each side of the bipartition and induced matchings $\MRS_1,\ldots,\MRS_t$ (this graph itself is known to both players). The existence 
of such RS graph is guaranteed by the results of~\cite{FischerLNRRS02} (see also~\cite{GoelKK12}). The hard distribution of the inputs is as follows; see~\Cref{fig:lower} for an illustration. 

\begin{tbox}
\textbf{A hard distribution $\GG$ of graphs.}

\begin{enumerate}
	\item Pick $\jstar \in [t]$ uniformly at random and let $\MRS_{\jstar}$ be the \emph{special} induced matching of $\GRS$. 
	\item For any vertex $v \in \GRS$, let $y_v = 1$ if $v \in V(\MRS_{\jstar})$ and $y_v = 0$ otherwise; sample a $k$-tuple $(x_{v,1},\ldots,x_{v,k})$ independently and uniformly at random conditioned on $x_{v,1} \oplus \cdots \oplus x_{v,k} = y_v$. 
	\item For any $v \in \GRS$, construct a vertex-disjoint XOR-gadget $\Gxor_v(x_{v,1},\ldots,x_{v,k})$ such that the final vertex of $\Gxor_v$ is the same as the vertex $v$. 
	\item For any edge $e \in \GRS$, drop $e$ from the graph independently and with probability half. Let $G$ be the resulting graph. 
\end{enumerate}
\end{tbox}

The distribution $\GG$ specifies the input graph $G$. The input to players is then determined by the distribution $\PP$ that sends each edge to one of the players chosen uniformly at random. 

\bigskip

\begin{figure}[h!]
\centering
\subcaptionbox{A graph $G$ sampled from $\GG$}%
  [0.45\linewidth]{

\begin{tikzpicture}
\tikzset{vertex/.style={circle, line width=1pt, blue, fill=blue!0,  draw, inner sep=0pt, minimum width=10pt, minimum height=10pt}}
\tikzset{vset/.style={ellipse, rounded corners, black, fill=black!0, draw, line width=1pt, inner sep=0pt, minimum width=24pt, minimum height=60pt}}

\node[vset] (L1){};
\node[vset] (L2) [below=1pt of L1]{};
\node[vset] (R1) [right=20pt of L1]{};
\node[vset] (R2) [below=1pt of R1]{};

\node[vertex] (l1) at ($(L1) + (0pt,18pt)$){};
\node[vertex] (l2) at ($(L1) + (0pt,0pt)$){};
\node[vertex] (l3) at ($(L1) + (0pt,-18pt)$){};

\node[vertex] (sl1) [left=50pt of l1]{};
\node[vertex] (sl2) [left=50pt of l2]{};
\node[vertex] (sl3) [left=50pt of l3]{};

\node[vertex] (x1) at ($(L2) + (0pt,18pt)$){};
\node[vertex] (x2) at ($(L2) + (0pt,0pt)$){};
\node[vertex] (x3) at ($(L2) + (0pt,-18pt)$){};

\node[vertex] (sx1) [left=50pt of x1]{};
\node[vertex] (sx2) [left=50pt of x2]{};
\node[vertex] (sx3) [left=50pt of x3]{};

\node[vertex] (r1) at ($(R1) + (0pt,18pt)$){};
\node[vertex] (r2) at ($(R1) + (0pt,0pt)$){};
\node[vertex] (r3) at ($(R1) + (0pt,-18pt)$){};

\node[vertex] (sr1) [right=50pt of r1]{};
\node[vertex] (sr2) [right=50pt of r2]{};
\node[vertex] (sr3) [right=50pt of r3]{};

\node[vertex] (y1) at ($(R2) + (0pt,18pt)$){};
\node[vertex] (y2) at ($(R2) + (0pt,0pt)$){};
\node[vertex] (y3) at ($(R2) + (0pt,-18pt)$){};

\node[vertex] (sy1) [right=50pt of y1]{};
\node[vertex] (sy2) [right=50pt of y2]{};
\node[vertex] (sy3) [right=50pt of y3]{};

\node[rectangle, ForestGreen, rounded corners, draw, inner sep=1pt,  fit=(l1) (sl1)] {};
\node[rectangle, ForestGreen, rounded corners, draw, inner sep=1pt,  fit=(l2) (sl2)] {};
\node[rectangle, ForestGreen, rounded corners, draw, inner sep=1pt,  fit=(l3) (sl3)] {};

\node[rectangle, ForestGreen, rounded corners, draw, inner sep=1pt,  fit=(x1) (sx1)] {};
\node[rectangle, ForestGreen, rounded corners, draw, inner sep=1pt,  fit=(x2) (sx2)] {};
\node[rectangle, ForestGreen, rounded corners, draw, inner sep=1pt,  fit=(x3) (sx3)] {};

\node[rectangle, ForestGreen, rounded corners, draw, inner sep=1pt,  fit=(r1) (sr1)] {};
\node[rectangle, ForestGreen, rounded corners, draw, inner sep=1pt,  fit=(r2) (sr2)] {};
\node[rectangle, ForestGreen, rounded corners, draw, inner sep=1pt,  fit=(r3) (sr3)] {};

\node[rectangle, ForestGreen, rounded corners, draw, inner sep=1pt,  fit=(y1) (sy1)] {};
\node[rectangle, ForestGreen, rounded corners, draw, inner sep=1pt,  fit=(y2) (sy2)] {};
\node[rectangle, ForestGreen, rounded corners, draw, inner sep=1pt,  fit=(y3) (sy3)] {};

\tikzset{decoration={snake,amplitude=.3mm,segment length=0.9mm,
                       post length=0mm,pre length=0mm}}
                       
\draw[color=black]
	(l1) to (r1)
	(l2) to (r2)
	(l3) to (r3) 
	
	(l1) to (y1)
	(l2) to (y2)
	(l3) to (y3) 
	
	(x1) to (r1)
	(x2) to (r2)
	(x3) to (r3) 
	
	(x1) to (y1)
	(x2) to (y2)
	(x3) to (y3);
	
\draw[color=ForestGreen, line width=1pt, decorate]
	(l1) to (sl1)
	(l2) to (sl2)
	(l3) to (sl3)
	
	(x1) to (sx1)
	(x2) to (sx2)
	(x3) to (sx3)
	
	(r1) to (sr1)
	(r2) to (sr2)
	(r3) to (sr3)
	
	(y1) to (sy1)
	(y2) to (sy2)
	(y3) to (sy3);

\end{tikzpicture}} 
 ~
\subcaptionbox{A maximum matching in $G$ }%
  [0.45\linewidth]{

\begin{tikzpicture}
\tikzset{vertex/.style={circle, blue, fill=blue!0,  line width=1pt, draw, inner sep=0pt, minimum width=10pt, minimum height=10pt}}
\tikzset{vset/.style={ellipse, rounded corners, black, fill=black!0, draw, line width=1pt, inner sep=0pt, minimum width=24pt, minimum height=60pt}}

\node[vset, fill=black!05] (L1){};
\node[vset] (L2) [below=1pt of L1]{};
\node[vset, fill=black!05] (R1) [right=20pt of L1]{};
\node[vset] (R2) [below=1pt of R1]{};

\node[vertex] (l1) at ($(L1) + (0pt,18pt)$){};
\node[vertex] (l2) at ($(L1) + (0pt,0pt)$){};
\node[vertex] (l3) at ($(L1) + (0pt,-18pt)$){};

\node[vertex] (sl1) [left=50pt of l1]{};
\node[vertex] (sl2) [left=50pt of l2]{};
\node[vertex] (sl3) [left=50pt of l3]{};

\node[vertex] (x1) at ($(L2) + (0pt,18pt)$){};
\node[vertex] (x2) at ($(L2) + (0pt,0pt)$){};
\node[vertex] (x3) at ($(L2) + (0pt,-18pt)$){};

\node[vertex] (sx1) [left=50pt of x1]{};
\node[vertex] (sx2) [left=50pt of x2]{};
\node[vertex] (sx3) [left=50pt of x3]{};

\node[vertex] (r1) at ($(R1) + (0pt,18pt)$){};
\node[vertex] (r2) at ($(R1) + (0pt,0pt)$){};
\node[vertex] (r3) at ($(R1) + (0pt,-18pt)$){};

\node[vertex] (sr1) [right=50pt of r1]{};
\node[vertex] (sr2) [right=50pt of r2]{};
\node[vertex] (sr3) [right=50pt of r3]{};

\node[vertex] (y1) at ($(R2) + (0pt,18pt)$){};
\node[vertex] (y2) at ($(R2) + (0pt,0pt)$){};
\node[vertex] (y3) at ($(R2) + (0pt,-18pt)$){};

\node[vertex] (sy1) [right=50pt of y1]{};
\node[vertex] (sy2) [right=50pt of y2]{};
\node[vertex] (sy3) [right=50pt of y3]{};

\node[rectangle, ForestGreen, rounded corners, draw, inner sep=1pt,  fit=(l1) (sl1)] {};
\node[rectangle, ForestGreen, rounded corners, draw, inner sep=1pt,  fit=(l2) (sl2)] {};
\node[rectangle, ForestGreen, rounded corners, draw, inner sep=1pt,  fit=(l3) (sl3)] {};

\node[rectangle, ForestGreen, rounded corners, draw, inner sep=1pt,  fit=(x1) (sx1)] {};
\node[rectangle, ForestGreen, rounded corners, draw, inner sep=1pt,  fit=(x2) (sx2)] {};
\node[rectangle, ForestGreen, rounded corners, draw, inner sep=1pt,  fit=(x3) (sx3)] {};

\node[rectangle, ForestGreen, rounded corners, draw, inner sep=1pt,  fit=(r1) (sr1)] {};
\node[rectangle, ForestGreen, rounded corners, draw, inner sep=1pt,  fit=(r2) (sr2)] {};
\node[rectangle, ForestGreen, rounded corners, draw, inner sep=1pt,  fit=(r3) (sr3)] {};

\node[rectangle, ForestGreen, rounded corners, draw, inner sep=1pt,  fit=(y1) (sy1)] {};
\node[rectangle, ForestGreen, rounded corners, draw, inner sep=1pt,  fit=(y2) (sy2)] {};
\node[rectangle, ForestGreen, rounded corners, draw, inner sep=1pt,  fit=(y3) (sy3)] {};

\tikzset{decoration={snake,amplitude=.4mm,segment length=2mm,
                       post length=0mm,pre length=0mm}}
                       
\draw[color=red, line width=1pt]
	(l1) to (r1)
	(l2) to (r2)
	(l3) to (r3);
	
\draw[color=ForestGreen, line width=1pt, dashed]
	(l1) to (sl1)
	(l2) to (sl2)
	(l3) to (sl3)

	(r1) to (sr1)
	(r2) to (sr2)
	(r3) to (sr3);
	
\draw[color=red, line width=1pt]
	(x1) to (sx1)
	(x2) to (sx2)
	(x3) to (sx3)
	
	(y1) to (sy1)
	(y2) to (sy2)
	(y3) to (sy3);

\end{tikzpicture}} 

\caption{An illustration of the distribution $\GG$ of input graphs and their maximum matchings. The middle graph is the ``base'' RS graph and each box connected to vertices of this RS graph denotes an XOR-gadget.}\label{fig:lower}
\end{figure}
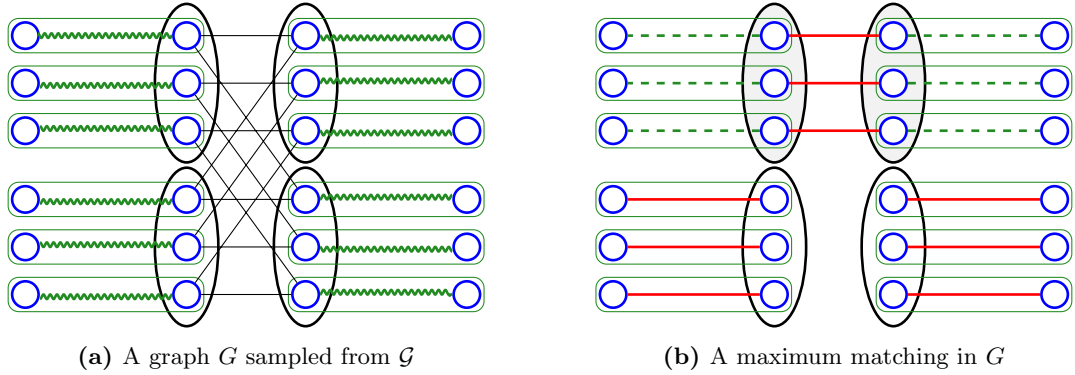

\bigskip

The following lemma specifies the key role of the special induced matching  in this distribution. 

\begin{lemma}\label{lem:lower-matching}
	For a graph $G \sim \GG$: 
	\begin{enumerate}[label=$(\roman*)$]
		\item $\expect{\mu(G)} \geq (N-r) \cdot 2k + 2r \cdot (k-1) + r/2$;
		\item$\mu(G \setminus \MRS_{\jstar}) \leq (N-r) \cdot 2k + 2r \cdot (k-1)$ with probability one;
	\end{enumerate}
\end{lemma}
\begin{proof}
	For this proof, it helps to refer to~\Cref{fig:lower} as a reference point. 
	
	Let us consider the graph $G \setminus \MRS_{\jstar}$ first. By~\Cref{lem:xor-gadget-match}, for every $v \in \GRS \setminus V(\MRS_{\jstar})$, $\Gxor_v(x_{v,1},\ldots,x_{v,k})$ has a matching of size $k$ since $x_{v,1} \oplus \ldots \oplus x_{v,k} = y_v = 0$
	in this case. The remaining XOR-gadgets also have a matching of size $k-1$ each, again by~\Cref{lem:xor-gadget-match} since now their XOR-values are $1$. Moreover, these latter matchings can be chosen so that 
	no vertex of $V(\MRS_{\jstar})$ is matched in them (by part $(ii)$ of~\Cref{lem:xor-gadget-match}).  Considering these matchings are 
	vertex-disjoint we have a matching $M$ of size $2 \cdot (N-r) \cdot k + 2r \cdot (k-1)$ in $G \setminus \MRS_{\jstar}$ with probability one that does not match any vertex of $\MRS_{\jstar}$. As a result: 
	
	$(i)$ In $G$, there is a matching consisting of $M$ plus all edges of $\MRS_{\jstar}$ present in $G$. As each of the edges of $\MRS_{\jstar}$ (with size $r$) is dropped w.p. half, we get the first part of the lemma. 
	
	$(ii)$ In $G \setminus \MRS_{\jstar}$, the matching $M$ is already a maximum matching. This is because, by the part $(i)$ of~\Cref{lem:xor-gadget-match}, the unique maximum matching of each XOR-gadget $\Gxor_v$ for 
	$v \in \GRS \setminus V(\MRS_{\jstar})$ necessarily matches $v$; hence, if we instead match $v$ to some vertex in $\GRS$, there will be one unmatched vertex in $\Gxor_v$ and thus size of the matching does not change. 
	As a result, the only vertices that can be matched inside $\GRS$ are $V(\MRS_{\jstar})$ but since $\MRS_{\jstar}$ consists of all edges between them (as $\MRS_{\jstar}$ is an \emph{induced} matching), 
	there is no edge left for these vertices in $G \setminus \MRS_{\jstar}$. 
\end{proof}

\subsection*{Auxiliary Random Variables and Input of Players} 
Let us now specify the random variables used in the distributions $\GG$ and $\PP$ explicitly: 
\begin{itemize}[leftmargin=20pt]
	\item $\rJ$ and $\rY := \set{\rY^v}$ for all $v \in \GRS$: the index $\jstar$ of the special matching $\MRS_{\jstar}$ and the corresponding random bits $y^v$ for XOR-gadgets. Notice that $\rJ$ and $\rY$ uniquely identify each other. 
	\item $\rX := \set{\rX^v := (\rX^v_1,\ldots,\rX^v_k)}$ for all $v \in \GRS$: the bits in XOR-gadgets of each vertex $v$ of $\GRS$. 
	\item $\rZ := \set{\rZ_e}$ for all $e \in \GRS$: $\rZ_e = 1$ for any $e \in \GRS$ that was chosen in $G$ and $\rZ_e = 0$ otherwise. 
	\item $\rPxor := \set{\rP_e}$ for all $e \in \Gxor_v$ among all $v \in \GRS$: $\rP_e = 1$ for any $e \in \Gxor_v$ that was sent to Alice as part of input  under the random partitioning $\PP$ and $\rP_e = 0$ otherwise.  
	\item $\rPrs := \set{\rP_e}$ for all $e \in \GRS$: $\rP_e = 1$ for any $e \in \GRS$ that was sent  to Alice as part of input  under the random partitioning $\PP$ and $\rP_e = 0$ otherwise.  Note that for technical reasons that will become evident shortly, 
	we have defined $\rPrs$ as partitioning all edges of $\GRS$ and not only the ones with $\rZ_e = 1$ that actually belong to the input graph. 
\end{itemize}
Additionally, we have the following definitions: 
\begin{itemize}[leftmargin=20pt]
	\item $\rX_A$ and $\rX_B$: we say that a bit $\rX^v_i$ is \emph{represented} in Alice's (resp. Bob's) input iff at least
	one of the edges $\Exor(\rX^v_i)$ is given to Alice (resp. Bob) by $\rPxor$ in partitioning of inputs (notice that $\rX^v_i$ might be represented in both players inputs); we use $\rX_A$ and $\rX_B$ to denote the bits 
	represented in Alice's and Bob's inputs, respectively. 
	\item $\rErsA$ and $\rErsB$: we say that $e \in \GRS$ is \emph{represented} in Alice's (resp. Bob's) input iff $\rP_e = 1$ (resp. $\rP_e = 0$), i.e., the partitioning $\rPrs$ assigns $e$ to Alice (resp. Bob); 
	we use $\rErsA$ and $\rErsB$ to denote the edges represented in Alice's and Bob's inputs, respectively (again, notice that by the definition of $\rPrs$ for all $e \in \GRS$, some edges are represented by Alice or Bob, but they may not belong to the graph $G$
	to begin with). 
	\item $\rMrsA(j)$ and $\rMrsB(j)$ for all $j \in [t]$: we define $\rMrsA(j)$ and $\rMrsB(j)$ analogously to $\rErsA$ and $\rErsB$ restricted to edges in  $\MRS_j$; so $\rErsA = (\rMrsA(1),\ldots,\rMrsA(t))$ and 
	$\rErsB = (\rErsB(1),\ldots,\rErsB(t))$. 
	\item $\rZA$ and $\rZB$: we define $\rZ_A := \set{\rZ_e}$ for $e \in \rErsA$ and $\rZ_B := \set{\rZ_e}$ for $e \in \rErsB$, that is the $\rZ$-values for edges represented in Alice's and Bob's inputs respectively. 
	Similarly, for any $j \in [t]$, we define $\rZ_A(j)$ and $\rZ_B(j)$ analogously to $\rZ_A$ and $\rZ_B$ restricted to edges in $\rMrsA(j)$ and $\rMrsB(j)$. 
\end{itemize}

We can now specify the input of Alice by the tuple $\rA := (\rPrs,\rZA,\rPxor,\rX_A)$ and input of Bob by $\rB:= (\rPrs,\rZB,\rPxor,\rX_B,\rJ)$. We note that these tuples are more general than the actual input of players. In particular, 
$\rPrs$ specifies the partitioning of edges that may not even be part of the input and Bob is explicitly given index $\rJ$; however, adding these more general inputs can only make our lower bounds stronger as Alice and Bob can always ignore this extra information. 

\paragraph{Hiding property of XOR-gadgets.} The key role of XOR-gadgets in our construction is that they ``hide'' the identity of the special induced matching $\MRS_{\jstar}$ \emph{from Alice}; we formalize this as follows. Define the following event: 
\begin{itemize}[leftmargin=15pt]
	\item \textbf{Event $\eventhide$:} for \emph{all} $v \in \GRS$, at least one of $(x^v_1,\ldots,x^v_k)$ is \emph{not} represented in Alice's input. 
\end{itemize}
We note $\eventhide$ is a deterministic function of the random variable $\rPxor$. We have, 

\begin{lemma}\label{lem:lower-hide}
	Suppose event $\eventhide$ happens. Then, even conditioned on the input $\rA$ of Alice, $\rJ$ is still chosen uniformly at random from $[t]$. 
\end{lemma}
\begin{proof}
	The only input of Alice which is, in principle, correlated with $\rJ$ is $\rX_A$; in general, $\rY$ and $\rJ$ uniquely identify each other and $\rX_A$ is used to determine $\rY$, namely, with a slight abuse of notation $\rY = \rX_A \oplus \rX_B$. 
	However, considering by $\eventhide$, $\rX_A$ ``misses'' at least one bit for \emph{every} 
	XOR-gadget, by~\Cref{lem:xor-gadget-hide}, any choice of $\rY$-value (even the \emph{correlated} ones obtained by picking $\rJ$) are equally likely conditioned on $\rX_A$, proving the lemma. 
\end{proof}

Finally, an easy calculation shows that $\eventhide$ happens with high probability. 
\begin{claim}\label{clm:lower-hide-prob}
	$\Pr\paren{\eventhide} \geq 1-o(1)$. 
\end{claim}
\begin{proof}
	Fix any vertex $v \in \GRS$. Any bit $x^v_i$ is represented in Alice's input if at least one of the two edges in $\Exor(x^v_i)$ is sent to Alice under $\PP$ which happens with probability $3/4$. As such, 
	the probability that $x^v$ is represented in Alice's input is only $(3/4)^{k} \leq 1/2N^{2}$ by the choice of $k$ in~\Cref{eq:lower-par}. A union bound on all $2N$ vertices in $\GRS$ finalizes the proof. 
\end{proof}

\subsection{Analysis of the Hard Distribution}

To start the analysis, we need to setup some notation. 

\paragraph{Notation.} Throughout this section, we fix a deterministic protocol $\prot$ over $\GG$ and random partitioning $\PP$ with communication cost $o(r \cdot t)$. We further 
let $\delta$ denote the probability that $\prot$ outputs an edge that does \emph{not} belong to $G$ (and thus errs). 
We use $\rProt$ to denote the random variable for the message $\Prot$ sent by Alice to Bob in $\prot$. We also use $\rMprot$ to denote the random variable for the matching output by the protocol $\prot$. 
Considering the input of Bob is $\rB$ and he additionally receives the message $\rProt$ from Alice, $\rMprot$ is a deterministic function of $(\rB,\rProt)$. 
In the following, $\en{\cdot}$ and $\mi{\cdot}{\cdot}$ denote the \emph{Shannon entropy} and \emph{mutual information}; see~\Cref{app:info} for more details. 

\smallskip

We first bound the size of $\rMprot$ based on the information revealed by $\rProt$ to Bob about edges of the special matching $\MRS_{\jstar}$ that are present in Alice's input, i.e., $\rZA(\rJ)$.  (In the following,   
$H_2$ is the binary entropy function, i.e., $H_2(\delta) := \en{\mathcal{B}(\delta)}$ where $\mathcal{B}(\delta)$ is a mean-$\delta$ Bernoulli random variable.)

\begin{lemma}\label{lem:lower-info}
	$\Ex\card{\rMprot} \leq (N-r) \cdot 2k + 2r \cdot (k-1) + r/4 + (1-H_2(\delta))^{-1} \cdot \mi{\rZA(\rJ)}{\rProt \mid \rB}.$ 
\end{lemma}
\begin{proof}
	By~\Cref{lem:lower-matching}, $\rMprot$ can only have $(N-r) \cdot 2k + 2r \cdot (k-1)$ edges outside of $\MRS_{\rJ}$. Hence, to prove the lemma, it suffices to 
	bound $\Ex\card{\rMprot \cap \MRS_{\rJ}}$. By definition, 
	\begin{align*}
	\Ex\card{\rMprot \cap \MRS_{\rJ}} &= \Ex\card{\rMprot \cap \rMrsA(\rJ)} + \Ex\card{\rMprot \cap \rMrsB(\rJ)} \leq  \Ex\card{\rMprot \cap \rMrsA(\rJ)} + r/4;
	\end{align*}
	this is because $\Ex{\card{\rMrsB(\rJ)}} = r/2$ (as each of the $r$ edges goes to Bob w.p. half) and among these, again, in expectation half of them belong to $G$, i.e., have $\rZ$-value $1$ (note that we can assume without loss of generality
	that Bob never outputs an edge $e \in \rMrsB(\rJ)$ with $\rZ_e=0$ as this edge is not part of input and thus makes the output wrong; moreover, unlike edges in Alice's input, here Bob directly knows $\rZ_e$ and can simply
	remove all edges with $\rZ_e=0$ from $\rMrsB(\rJ)$).  
	
	To finalize the proof, we need to show 
	\begin{align}
		\Ex\card{\rMprot \cap \rMrsA(\rJ)} \leq (1-H_2(\delta))^{-1} \cdot \mi{\rZA(\rJ)}{\rProt \mid \rB}. \label{eq:lower-info-eq-1}
	\end{align}
	Let us condition on any choice of $\rPrs=P$ and $\rJ = j$ in $\rB = (\rPrs,\rZB,\rPxor,\rX_B,\rJ)$. This fixes $\rMrsA(\rJ)$ to some matching $M_A(j) \subseteq \MRS_j$,  
	but  $\set{\rZ_e}$ for $e \in M_A(j)$ are still uniformly distributed as $\rZ \perp \rPrs,\rJ$. Fix any edge $e \in M_A(j)$. 
	For Bob to be able to output $e$ as part of $\rM_{\prot}$, 
	the entropy of $\rZ_e$ should be sufficiently small conditioned on $(\rB,\rProt)$; otherwise, Bob is likely to output an edge that does not belong to the graph and thus errs. 
	Formally, for any $e \in \rMprot \cap M_A(j)$, $\Pr\paren{\rZ_e = 0 \mid \rProt, \rB} \leq \delta$ which implies that, 
	\begin{align}
		\en{\rZ_e \mid \rProt, \rB} \leq H_2(\delta),  \label{eq:lower-info-eq-2}
	\end{align}

	We are going to use this to bound the information revealed about $\rZA(\rJ)$ by Alice's message. Let $L := L(P,j)$ denote the set of ``low entropy'' edges in $M_A(j)$, i.e., 
	all edges $e \in M_A(j)$ that satisfy~\Cref{eq:lower-info-eq-2} conditioned on $\rPrs=P$ and $\rJ = j$. As discussed, 
	\begin{align}
	\Ex\card{\rMprot \cap \rMrsA(\rJ)} \leq \Ex_{P,j}{\card{L(P,j)}}. \label{eq:lower-info-eq-3}
	\end{align}
	We now bound the RHS above as follows. By the definition of $\rB = (\rPrs,\rZB,\rPxor,\rX_B,\rJ)$, 
	\begin{align*}
		\mi{\rZA(\rJ)}{\rProt \mid \rB} &= \Ex_{P,j} \bracket{\mi{\rZA(j)}{\rProt \mid \rPrs=P,\rZB,\rPxor,\rX_B,\rJ=j}} \\
		&= \Ex_{P,j} \Bracket{\en{\rZA(j) \mid \rPrs=P,\rZB,\rPxor,\rX_B,\rJ=j} - \en{\rZA(j) \mid \rProt , \rPrs=P,\rZB,\rPxor,\rX_B,\rJ=j}} \\ 
		&= \Ex_{P,j} \Bracket{\card{M_A(j)} - \en{\rZA(j) \mid \rProt , \rPrs=P,\rZB,\rPxor,\rX_B,\rJ=j}} \tag{as $\set{\rZ_e}$ for $e \in M_A(j)$ is uniformly distributed conditioned on the remaining variables} \\
		&\geq  \Ex_{P,j} \Bracket{\card{M_A(j)} - \sum_{e \in M_A(j)} \en{\rZ_e \mid \rProt , \rPrs=P,\rZB,\rPxor,\rX_B,\rJ=j}} \tag{by the sub-additivity of entropy} \\
		&\geq  \Ex_{P,j} \Bracket{\card{M_A(j)} - (\card{M_A(j)}-\card{L(P,j)}+ \sum_{e \in L(P,j)} \en{\rZ_e \mid \rProt , \rPrs=P,\rZB,\rPxor,\rX_B,\rJ=j})} \tag{by upper bounding the entropy of the terms not in $L(P,j)$ by one} \\
		&\geq  \Ex_{P,j} \Bracket{\card{M_A(j)} - (\card{M_A(j)}-\card{L(P,j)}+ \sum_{e \in L(P,j)} H_2(\delta))} \tag{by the definition of $L(P,j)$ based on~\Cref{eq:lower-info-eq-2}} \\
		&= \paren{1-H_2(\delta)} \cdot \Ex_{P,j} \card{L(P,j)}.
	\end{align*}
	Plugging in this bound in~\Cref{eq:lower-info-eq-2} finalizes the proof. 
\end{proof}

The main part of the proof is to bound the mutual information term in the RHS of~\Cref{lem:lower-info}, i.e., show that a low communication protocol cannot reveal much information 
about $\rZ(\rJ)$ even conditioned on all the inputs of Bob. 

\begin{lemma}\label{lem:lower-info-bound}
	$\mi{\rZ_A(\rJ)}{\rProt \mid \rB} = o(r)$. 
\end{lemma}
\begin{proof}
	Recall that $\rB = (\rPrs,\rZB,\rPxor,\rX_B,\rJ)$ and that any choice $P$ for $\rPxor$, determines whether or not the event $\eventhide$ happens when for $\rPxor = P$. As such, 
	\begin{align}
		\mi{\rZ_A(\rJ)}{\rProt \mid \rB} &= \Ex_{P} \,\bracket{\mi{\rZ_A(\rJ)}{\rProt \mid \rPrs,\rZB,\rPxor=P,\rX_B,\rJ}} \notag \\
		&\leq  \Ex_{P \mid \eventhide} \,\bracket{\mi{\rZ_A(\rJ)}{\rProt \mid \rPrs,\rZB,\rPxor=P,\rX_B,\rJ}} + (1-\Pr\paren{\eventhide}) \cdot r \tag{as this mutual information term can be at most $r$}  \\
		&= \Ex_{P \mid \eventhide} \,\bracket{\mi{\rZ_A(\rJ)}{\rProt \mid \rPrs,\rZB,\rPxor=P,\rX_B,\rJ}} + o(r), \label{eq:lower-info-bound-eq1}
	\end{align}
	where the final step is by~\Cref{clm:lower-hide-prob}. 
	
	We now focus only on the cases when $\eventhide$ happens in the RHS above. Choosing a value $P'$ for $\rPrs$ determines $\rMrsA(1),\ldots,\rMrsA(t)$ and the partitioning of $\rZ$ into $\rZA$ and $\rZB$. Thus, we have, 
	\begin{align}
		\text{First term in the RHS of~\eqref{eq:lower-info-bound-eq1}} &= \Ex_{P,P' \mid \eventhide} \,\bracket{\mi{\rZ_A(\rJ)}{\rProt \mid \rPrs=P',\rZB,\rPxor=P,\rX_B,\rJ}} \notag \\
				 &\leq \Ex_{P,P' \mid \eventhide} \,\bracket{\mi{\rZ_A(\rJ)}{\rProt \mid \rPrs=P', \rPxor=P,\rX_B,\rJ}}; \label{eq:lower-info-bound-eq2}
	\end{align}
	this is because, the input of Alice conditioned on $\rB$ is determined only by $\rZA$ and $\rX_A$ and both these variables are independent of $\rZB$, which implies, 	$\rProt \perp \rZB \mid \rZ_A(\rJ), \rPrs=P',\rZB,\rPxor=P,\rX_B,\rJ$
	and thus we can apply~\Cref{prop:info-decrease} to remove conditioning on $\rZB$. 
	
	Our goal now is to also remove the conditioning on $\rX_B$. However, this is not as direct as the previous step as $(\rX_B,\rJ)$ together are correlated with the input of Alice (in particular, $\rX_A$) and we cannot use the previous argument. 
	Instead, we are going to show that we can in fact ``switch'' $\rX_B$ with $\rX_A$ in the conditioning above without decreasing the RHS. We claim that, for any $P,P'$, 
	\begin{align*}
		\mi{\rZ_A(\rJ)}{\rProt \mid \rPrs=P', \rPxor=P,\rX_B,\rJ} \leq \mi{\rZ_A(\rJ)}{\rProt \mid \rPrs=P', \rPxor=P,\rX_B,\rJ,\rX_A};
	\end{align*}
	this is because $\rZ_A(\rJ) \perp \rX_A \mid \rPrs=P', \rPxor=P,\rX_B,\rJ$ as $\rZ$-values and $\rX$-values are chosen independently (and none of the conditions correlate them); thus we can apply~\Cref{prop:info-increase}. 
	We can now remove $\rX_B$ from the conditioning: 
	\begin{align*}
		\mi{\rZ_A(\rJ)}{\rProt \mid \rPrs=P', \rPxor=P,\rX_B,\rJ,\rX_A} \leq \mi{\rZ_A(\rJ)}{\rProt \mid \rPrs=P', \rPxor=P,\rJ,\rX_A};
	\end{align*}
	this is because $\rProt \perp \rX_B \mid \rZ_A(\rJ), \rPrs=P', \rPxor=P,\rJ,\rX_A$ as $\rProt$ is only a function of $\rX_A$ and $\rZA$ after the conditioning, and in particular is independent of $\rX_B$; thus we can apply~\Cref{prop:info-decrease}.  
	
	By plugging in these bounds in the RHS of~\Cref{eq:lower-info-bound-eq2}, we obtain that, 
	\begin{align*}
		\text{RHS of~\eqref{eq:lower-info-bound-eq2}} &\leq \Ex_{P,P' \mid \eventhide} \,\bracket{\mi{\rZ_A(\rJ)}{\rProt \mid \rPrs=P', \rPxor=P,\rJ,\rX_A}} \\
		&=  \Ex_{P,P' \mid \eventhide} \,\bracket{\sum_{j=1}^{t} \Pr\paren{\rJ=j \mid P,P'} \cdot \mi{\rZ_A(j)}{\rProt \mid \rPrs=P', \rPxor=P,\rJ=j,\rX_A}} \\
		&= \frac{1}{t} \cdot \Ex_{P,P' \mid \eventhide} \,\bracket{\sum_{j=1}^{t} \mi{\rZ_A(j)}{\rProt \mid \rPrs=P', \rPxor=P,\rJ=j,\rX_A}} \tag{as $\rJ \perp \rPrs,\rPxor$ and is uniform over $[t]$} \\
		&=  \frac{1}{t} \cdot \Ex_{P,P' \mid \eventhide} \,\bracket{\sum_{j=1}^{t} \mi{\rZ_A(j)}{\rProt \mid \rPrs=P', \rPxor=P,\rX_A}}; 
	\end{align*}
	in the last step, we can drop the conditioning on the event $\rJ=j$ as the joint distribution of the remaining variables $(\rZ_A(j),\rProt,\rX_A)$ is independent of $\rJ$: this is because
	all these variables only depend on the input of Alice, while conditioned on the event $\eventhide$, by~\Cref{lem:lower-hide}, the input of Alice is independent of $\rJ$. 
	
	We can continue the above calculations as follows: 
	\begin{align*}
		\text{RHS of~\eqref{eq:lower-info-bound-eq2}} &\leq \frac{1}{t} \cdot \Ex_{P,P' \mid \eventhide} \,\bracket{\sum_{j=1}^{t} \mi{\rZ_A(j)}{\rProt \mid \rPrs=P', \rPxor=P,\rX_A}} \\
		&\leq \frac{1}{t} \cdot \Ex_{P,P' \mid \eventhide} \,\bracket{\sum_{j=1}^{t} \mi{\rZ_A(j)}{\rProt \mid \rPrs=P', \rPxor=P,\rX_A, \rZ_A([1:j-1])}}  
		\tag{by~\Cref{prop:info-increase} as $\rZ_A(j) \perp \rZ_A([1:j-1]) \mid \rPrs=P', \rPxor=P,\rX_A$}  \\
		&=  \frac{1}{t} \cdot \Ex_{P,P' \mid \eventhide} \,\bracket{\mi{\rZ_A}{\rProt \mid \rPrs=P', \rPxor=P,\rX_A}}  \tag{by the chain rule of mutual information (\itfacts{chain-rule})} \\
		&= \frac{1}{t} \cdot \bracket{\mi{\rZ_A}{\rProt \mid \rPrs, \rPxor,\rX_A}} \leq \frac{1}{t} \cdot \en{\rProt} \leq o(r) \tag{as $\en{\rProt} = o(r \cdot t)$ since $\prot$ only communicates $o(r \cdot t)$ many bits and by~\itfacts{uniform}}.
	\end{align*}
	Plugging in this bound in~\Cref{eq:lower-info-bound-eq2} and then in turn in~\Cref{eq:lower-info-bound-eq1} finalizes the proof. 
\end{proof}

Suppose the error probability of the protocol $\prot$, i.e., $\delta$, is some constant bounded away from zero. Then, \Cref{lem:lower-info,lem:lower-info-bound}, together with the fact that $H_2(\delta) < 1$,
 imply the following upper bound on the size of the matching output by Bob: 
\begin{align*}
	\Ex\card{\rMprot} \leq (N-r) \cdot 2k + 2r \cdot (k-1) + r/4 + o(r) = \Ex\bracket{\mu(G)} - r/2 + o(r), 
\end{align*}
where the equality is by part $(i)$ of~\Cref{lem:lower-matching}. On the other hand, since $\Ex\bracket{\mu(G)} \leq (2k+1) \cdot N$ (as number of vertices is twice this quantity), we have 
that     
\begin{align*}
	\frac{\Ex\card{\rMprot}}{\Ex\bracket{\mu(G)}} \leq 1 - \frac{r/2+o(r)}{(2k+1) \cdot N} := 1-\eps_0, 
\end{align*}
for some $\eps_0 = \Theta(\nicefrac{1}{\log{N}})$ (as $r=N/3$ and $k = \Theta(\log{N})$). 

Finally, note that the number of vertices in the graph is $n = (2k+1) \cdot N = \Theta(N\cdot\log{N})$. As such, we obtain that 
any deterministic protocol with communication cost $o(n^{1+\Theta(\nicefrac{1}{\log\log{n}})})$ which is $o(r \cdot t) = o(N^{1+\Theta(\nicefrac{1}{\log\log{N}})})$ (for an appropriate choice of constant in the exponent), 
cannot output a $(1-\eps_0)$-approximation to maximum matching in expectation (even if it is allowed to err with constant probability by outputting an edge not in the graph). Moreover, we can immediately 
extend this result to randomized protocols using the ``easy direction'' of Yao's minimax principle (i.e., an averaging argument). 

This concludes the proof of~\Cref{thm:lower-stream} by the connection between communication complexity and streaming lower bounds.

\begin{remark}
	For the simplicity of exposition, we compared the \underline{expected} size of the matching of the protocol vs a maximum matching of the input. However, a simple application of Markov bound also 
	imply the same result for protocols that output a $(1-\eps_0)$-approximate matching with any constant probability of success. 
	
	 Basically, the lower bound for $\mu(G)$ in~\Cref{lem:lower-matching} is highly concentrated (a simple Chernoff bound on edges of the special matching that belong to the graph). 
	 Also, the only term in our upper bound of $\rMprot$ in~\Cref{lem:lower-info} which is not necessarily concentrated is 
	 the mutual information term which is only $o(r)$ by~\Cref{lem:lower-info-bound}; hence, by Markov bound, $\card{\rMprot} \leq {\mu(G)} - r/2 + o(r)$ with probability $1-o(1)$ and not only in expectation. 
\end{remark}

\subsection*{Acknowledgements}
We thank Aaron Bernstein for helpful conversations on the random-order streaming matching problem and several insightful comments that  helped us in improving the presentation of the paper. 

\bibliographystyle{alpha}
\bibliography{general}

\appendix

\newcommand{\rU}{\rv{U}}
\newcommand{\unif}{\mathcal{U}}

\newcommand{\rT}{\rv{T}}

\clearpage

\section{Tools from Information Theory}\label{app:info}

We shall use the following basic properties of entropy and mutual information throughout; the proofs can be found in~\cite[Chapter~2]{CoverT06}. 

\begin{fact}\label{fact:it-facts}
  Let $\rA$, $\rB$, $\rC$, and $\rD$ be four (possibly correlated) random variables.
   \begin{enumerate}
  \item \label{part:uniform} $0 \leq \en{\rA} \leq \log{\card{\supp{\rA}}}$. The right equality holds
    iff $\distribution{\rA}$ is uniform.
  \item \label{part:info-zero} $\mi{\rA}{\rB} \geq 0$. The equality holds iff $\rA$ and
    $\rB$ are \emph{independent}.
  \item \label{part:cond-reduce} \emph{Conditioning on a random variable reduces entropy}:
    $\en{\rA \mid \rB,\rC} \leq \en{\rA \mid  \rB}$.  The equality holds iff $\rA \perp \rC \mid \rB$.
    \item \label{part:sub-additivity} \emph{Subadditivity of entropy}: $\en{\rA,\rB \mid \rC}
    \leq \en{\rA \mid \rC} + \en{\rB \mid  \rC}$.
  \item \label{part:chain-rule} \emph{Chain rule for mutual information}: $\mi{\rA,\rB}{\rC \mid \rD} = \mi{\rA}{\rC \mid \rD} + \mi{\rB}{\rC \mid  \rA,\rD}$.
   \end{enumerate}
\end{fact}

\noindent
We also use the following two standard propositions. 

\begin{proposition}\label{prop:info-increase}
  For random variables $\rA, \rB, \rC, \rD$, if $\rA \perp \rD \mid \rC$, then, 
  \[\mi{\rA}{\rB \mid \rC} \leq \mi{\rA}{\rB \mid  \rC,  \rD}.\]
\end{proposition}
 \begin{proof}
  Since $\rA$ and $\rD$ are independent conditioned on $\rC$, by
  \itfacts{cond-reduce}, $\HH(\rA \mid  \rC) = \HH(\rA \mid \rC, \rD)$ and $\HH(\rA \mid  \rC, \rB) \ge \HH(\rA \mid  \rC, \rB, \rD)$.  We have,
	 \begin{align*}
	  \mi{\rA}{\rB \mid  \rC} &= \HH(\rA \mid \rC) - \HH(\rA \mid \rC, \rB) = \HH(\rA \mid  \rC, \rD) - \HH(\rA \mid \rC, \rB) \\
	  &\leq \HH(\rA \mid \rC, \rD) - \HH(\rA \mid \rC, \rB, \rD) = \mi{\rA}{\rB \mid \rC, \rD}. \qed
	\end{align*} 
	
\end{proof}

\begin{proposition}\label{prop:info-decrease}
  For random variables $\rA, \rB, \rC,\rD$, if $ \rA \perp \rD \mid \rB,\rC$, then, 
  \[\mi{\rA}{\rB \mid \rC} \geq \mi{\rA}{\rB \mid \rC, \rD}.\]
\end{proposition}
 \begin{proof}
 Since $\rA \perp \rD \mid \rB,\rC$, by \itfacts{cond-reduce}, $\HH(\rA \mid \rB,\rC) = \HH(\rA \mid \rB,\rC,\rD)$. Moreover, since conditioning can only reduce the entropy (again by \itfacts{cond-reduce}), 
  \begin{align*}
 	\mi{\rA}{\rB \mid  \rC} &= \HH(\rA \mid \rC) - \HH(\rA \mid \rB,\rC) \geq \HH(\rA \mid \rD,\rC) - \HH(\rA \mid \rB,\rC) \\
	&= \HH(\rA \mid \rD,\rC) - \HH(\rA \mid \rB,\rC,\rD) = \mi{\rA}{\rB \mid \rC,\rD}. \qed
 \end{align*}
 
\end{proof}

\end{document}